%% file: rosa25_hal.tex
\documentclass[11pt]{article}

\usepackage[a4paper,margin=2.5cm]{geometry}
\usepackage{graphicx}
\usepackage{bm}
\usepackage{xcolor}
\usepackage{mathrsfs}
\usepackage{enumitem}
\usepackage{algorithm}
\usepackage{float}
\usepackage{xspace}
\usepackage{comment}
\usepackage{algpseudocode}
\usepackage{amsmath}
\usepackage{amssymb}
\usepackage{amsthm}
\usepackage{balance}
\usepackage{hyperref}
\usepackage[capitalize]{cleveref} 
\usepackage{microtype}

\allowdisplaybreaks           
\theoremstyle{plain}
\newtheorem{theorem}{Theorem}[section]
\newtheorem{lemma}[theorem]{Lemma}
\newtheorem{proposition}[theorem]{Proposition}
\newtheorem{corollary}[theorem]{Corollary}
\theoremstyle{definition}
\newtheorem{definition}[theorem]{Definition}
\theoremstyle{remark}

\numberwithin{equation}{section}

\newcommand{\Z}{{\mathbb{Z}}} 
\newcommand{\R}{{\mathbb{R}}} 
\newcommand{\Q}{{\mathbb{Q}}} 
\newcommand{\C}{{\mathbb{C}}} 

\newcommand{\GL}{{\mathrm{GL}}} 

\newcommand{\bmM}{{\bm{M}}} 
\newcommand{\bmA}{{\bm{A}}} 

\newcommand{\bmx}{{\bm{x}}} 
\newcommand{\bmc}{{\bm{c}}} 
\newcommand{\bma}{{\bm{a}}} 
\newcommand{\bmg}{{\bm{g}}}
\newcommand{\bmm}{{\bm{m}}}
\newcommand{\bmf}{{\bm{f}}}

\newcommand{\calC}{{\mathcal{C}}} 
\newcommand{\calS}{{\mathcal{S}}} 
\newcommand{\calP}{{\mathcal{P}}}

\newcommand{\calH}{{\mathcal{H}}}
\newcommand{\calV}{{\mathcal{V}}}
\newcommand{\calE}{{\mathcal{E}}}

\newcommand{\dep}{{\delta_\calP}}
\newcommand{\tp}{{\tau_\calP}}
\newcommand{\de}{\delta}

\newcommand{\frakP}{{\mathfrak{P}}}
\newcommand{\frakA}{{\mathfrak{A}}} 
\newcommand{\fraka}{{\mathfrak{a}}} 
\newcommand{\frakc}{{\mathfrak{c}}} 

\newcommand{\scrK}{{\mathscr{K}}} 
\newcommand{\sing}{{\textrm{sing}}}
\newcommand{\reg}{{\textrm{reg}}}
\newcommand{\app}{{\mathrm{app}}}
\newcommand{\heightP}{{ \softO{n^6d^{n+1} + n^8d^n(h+\log\left(
\frac{1}{\epsilon} \right) }}}
\newcommand{\curve}{{\calC}} 
\newcommand{\softO}[1]{\mathchoice{\tilde{O}\left(#1\right)}{O\tilde{~}(#1)}{O\tilde{~}(#1)}{O\tilde{~}(#1)}} 

\newcommand{\hypS}{{(\mathsf{S})}}

\newcommand{\PartialComponents}{{\textsc{PartialComponents}}\xspace}
\newcommand{\RegularPoints}{{\textsc{RegularPoints}}\xspace}
\newcommand{\ConnectParams}{{\textsc{ConnectParams}}\xspace}
\newcommand{\PlanarComponents}{{\textsc{PlanarComponents}}\xspace}
\newcommand{\GenPlanarComponents}{{\textsc{GenPlanarComponents}}\xspace}
\newcommand{\ComputeRealComponents}{{\textsc{CurveComponents}}\xspace}

\newcommand{\querypoints}{{\calP}}
\newcommand{\queryparam}{{\mathfrak{P}}}
\newcommand{\curveparam}{{\mathfrak{R}}}
\newcommand{\Zeroes}{{\mathsf{Z}}}

\newcommand{\compprecond}{{\softO{\delta^6\tau + \delta^7}}}

\newcommand{\compconnectparametrizations}{{O\left( \delta^4+\delta \delta_\mathcal{T}\right)}}

\newcommand{\compsimpplanargraph}{{\softO{(\delta^6+\delta^4\delta_\querypoints+\delta
      \delta_{\querypoints}^2)
(\tau+\delta)+(\delta^5+\delta^3\delta_\querypoints+\delta_\querypoints^2)
\tau_\querypoints+\delta_\calP^3}}}
\newcommand{\complgenplanarcomponents}{{\widetilde{O}((\delta^6+\delta^4\delta_\querypoints+\delta
    \delta_{\querypoints}^2)
(\tau+\delta)+(\delta^5+\delta^3\delta_\querypoints+\delta_\querypoints^2)
\tau_\querypoints+\delta_\querypoints^3+\delta_q^2\tau_q)}}
\newcommand{\compP}{{
    \softO{n^{14}d^{2n+2}\binom{n+d}{n}\left( h+\log(\frac{1}{\epsilon})
\right)}
}}
\newcommand{\totalcompP}{{\softO{n^{15}d^{2n+2}\binom{n+d}{n}\left( h+\log(\frac{1}{\epsilon})
\right)}}}

\newcommand{\totalcomplspace}{
\[
\begin{aligned}
\widetilde{O}\!\Bigl(
  &n d^{7n-7} h + n^6d^{6n-4}+n^9d^{5n-2}+n^{12}d^{3n+1}+ \\[-0.5em]
  &\bigl(d^{6n-6}+n^4 d^{5n-4}+n^7 d^{3n-1}\bigr)
    \overline{\tau}+
    \!(h+\log\!\frac{1}{\varepsilon})\times 
    \\[-0.5em]
  &\Bigl(
      n^8\!\bigl(d^{6n-5}+n^3 d^{5n-3}+n^6 d^{3n}\bigr)
      + \binom{n+d}{n}\!\bigl(n^4 d^{3n}+n^{16} d^{2n+4}\bigr)
    \Bigr)
    \Bigr).
\end{aligned}
\]
}
\newcommand{\simplifiedtotalcompspace}{{\widetilde{O}(n^{12}d^{7n-7}(h+\overline{\tau}+1)+
(h+\log \frac{1}{\varepsilon})(n^{14} d^{6n-5}+\binom{n+d}{n} n^{16} d^{3n}))}}

\newcommand{\componedimparam}{{\softO{(h+\log(\frac{1}{\epsilon}))\binom{n+d}{n}n^4d^{3n}}}}
\newcommand{\heightonedimparam}{{\softO{d^{n-1}(nh+n)}}}

\title{Computing the Connected Components of Real Algebraic Curves}

\newcommand{\authorspacing}{1.8cm} 
\author{%
  \begin{tabular}{@{}c@{\hspace{\authorspacing}}c@{}}
    Elisabetta Rocchi & Mohab Safey El Din \\
    \small\texttt{Elisabetta.Rocchi@lip6.fr} & \small\texttt{Mohab.Safey@lip6.fr}
  \end{tabular}\\[1em]
  \small Sorbonne Universit\'e, LIP6, CNRS, UMR7606, Paris, France
}

\date{}

\begin{document}

\maketitle

\begin{abstract}
Connected components of real algebraic sets are semi-algebraic, i.e.
they are described by a boolean formula whose atoms are polynomial
constraints with real coefficients. Computing such descriptions finds
topical applications in optical system design and robotics. In this
paper, we design a new algorithm for computing such semi-algebraic
descriptions for real algebraic curves. Notably, its complexity is
less than the best known one for computing a graph which is isotopic
to the real space curve under study.
\end{abstract}

\input{rosa25_source}

\end{document}

%% file: rosa25_source.tex

\section[\NoCaseChange{Introduction}]{\texorpdfstring{\NoCaseChange{Introduction}}{Introduction}}

\paragraph*{Problem statement.} Connected components of real algebraic
sets are \emph{semi-algebraic sets} and then described 
by a boolean formula of polynomial constraints with
real coefficients. "Computing the connected components" 
means computing such descriptions. \\
This paper is devoted to the computation of
the connected components of real algebraic curves of $\R^n$. 

\paragraph*{Prior works and motivation.}
This is motivated by optical applications where 
an optical co-axial design
classification problem \cite{LBVR21, Drogoul25}
boils down to compute the connected components of space curves. 
A first approach is to combine the cylindrical algebraic decomposition
(CAD) algorithm \cite{Collins} with algorithms establishing adjacency
relations between cells computed by CAD such as
\cite{McCaCo02,ArCoMcCa88, ArCoMcCa98}, or \cite{Strz17}. 
This is doubly exponential in $n$ \cite{BD07}. 

\cite[Chap. 16]{BPR} shows that one can compute the connected
components of semi-algebraic
sets defined by $s$ polynomials of degree at most $d$, 
using $s^{n+1}d^{O\left( n^4 \right)}$ arithmetic operations 
(see also \cite{CGV92,
HRS94}).

For sets of dimension one, starting with re\-al
algebraic cur\-ves, one might expect better 
results. Many cont\-ributions 
"compute the topology" of plane curves (e.g.~\cite{AlMoWi08,
ArCa88, BerberichEmeliyanenkoKobelSagraloff2013, BuChoGaYa08, CheWu18, CheLaPePoRoTsi10, EiKe08,
EiKeWo07, GoNe02, Hong96, KoSa15,
DiattaDiattaRouillierRoySagraloff2022, SeiWo05}) and space curves in
$\R^3$ (e.g.~\cite{AlSe05, CheJiLa13, DaMoRu08, ElKahoui08,
GaLaMoTe05, JinChe20}). 
"Computing the topology" means here computing an embedded
graph, isotopic to the curve under study. For $f \in \Z[\bmx]$ (with
$\bmx = x_1, \ldots, x_n$),
let $(d, \tau)$ be its magnitude, i.e. $d$ and $\tau$ are 
its degree and the maximum bit size of its coefficients. 
As far as we know, the best complexity for space curves in
$\R^3$ is achieved by the deterministic algorithm in
\cite{cheng:hal-04874978} which uses $\softO{ d^{18} +
d^{17}\tau }$ bit operations where $(d, \tau)$ bounds the magnitudes of
the input polynomials ($\softO{.}$ notation
means that poly-logarithmic factors are omitted). \\ 
Counting the
connected components of a given curve or answering connectivity
queries might be cheaper. To go further, we
recall now how such query points and curves can be encoded. \\
A zero-dimensional parametrization $\queryparam$ 
is $\lambda, (w, v_1,
\ldots, v_n)$ in $\Q[\bmx]\times \Q[t]^{n+1}$ with
$\lambda$ is a linear form, $w$ is square-free and monic. It defines
the set $\Zeroes(\queryparam) = 
\left \{\left( \frac{v_1}{\partial w / \partial t},
    \ldots, \frac{v_n}{\partial w / \partial t} \right)(\theta) \mid
  w(\theta) = 0\right \} $. All finite sets of $\Q$-algebraic points
  can be encoded by a zero-dimensional parametrization. This is how
  query points are encoded.\\
A one-dimensional parametrization $\curveparam = (\lambda, \mu, (w, v_1,
\ldots, v_n)) \in \Q[\bmx]^2\times \Q[t,u]^{n+1}$ with
$\lambda, \mu$ being linear forms, and $w$ be square-free and monic in $t$ and
$u$ defines the curve $\Zeroes(\curveparam)$ which is the Zariski closure
of $\left \{\left( \frac{v_1}{\partial w / \partial u},
    \ldots, \frac{v_n}{\partial w / \partial u} \right)(\theta, \eta) \mid
  w(\theta, \eta) = 0, \frac{\partial w}{\partial t}(\theta, \eta)\neq
0\right \} $. All $\Q$-algebraic curves can be
defined by such a parametrization. 

\cite{IslamPoteauxPrebet2023} shows that, given a one-dimensional
parametrization of magnitude $\left(\delta, \tau \right)$ encoding a
curve $\curve$ in $\C^n$  and a zero-dimensional parametrization of
magnitude $(\mu, \kappa)$ encoding $\calP\subset \curve$, the connected
components of $\curve\cap\R^n$ and connectivity queries involving
$\calP\cap\R^n$ can be counted/answered using $\softO{ \delta^{6}
+\mu^{6} +\delta^{5}\tau +\mu^{5}\kappa}$ bit operations. For curves
of $\R^3$ defined as the intersection of two polynomials of maximum
magnitude $(d, \tau)$, this yields a bit complexity $\softO{
d^{12}(1+\tau) }$, hence less than the one needed by
\cite{cheng:hal-04874978}. \\
Still, there is still a lack of algorithms with comparable
complexities for computing the connected components of real algebraic
curves.

\paragraph*{Main contributions.} 
We design an algorithm named \ComputeRealComponents for computing the
connected components of a real algebraic curve in $\R^n$ given as the
solution set to a system of polynomial equations $f_1 = \cdots
=f_{n-1}=0$ with coefficients in $\Q$ and generating a radical ideal,
whose associated algebraic curve is denoted by $\curve$. 
\\
This algorithm builds upon \cite{IslamPoteauxPrebet2023} which shows
how to analyze the connectivity of $\curve_\R = \curve\cap \R^n$,
from the analysis of the topology of the projection of $\curve$ on a
generic enough plane. Hence, a first change of coordinates induced by
$\bmA_1\in \GL_n(\Q)$ is performed to ensure genericity properties
(called $(H)$ in \cite{IslamPoteauxPrebet2023}) and obtain a
one-dimensional parametrization $\curveparam_1$ for $\curve^{\bmA_1} =
\{\bmA_1^{-1}x \mid x \in \curve\} $. In that situation, the
projection $\curve_{2, \bmA_1}$ of $\curve^{\bmA_1}$ on the first two
coordinates is an algebraic curve. As shown in
\cite{IslamPoteauxPrebet2023}, the connectivity properties of
$\curve_{2, \bmA_1}\cap \R^2$ can be easily lifted to
$\R^n$ except at the set of apparent singularities $\app\left(
\curve_{2, \bmA_1} \right)$ of
$\curve_{2,\bmA_1}$ which is the set of singularities of $\curve_{2,
\bmA_1}$ outside the projection of the singular points of
$\curve^{\bmA_1}$. One of the tour de force in
\cite{IslamPoteauxPrebet2023} is to show how to lift connectivity
properties to $\R^n$ when passing through $\app(\curve_{2,
\bmA_1})\cap \R^2$.
\\
Our first auxiliary contribution is an
algorithm named \GenPlanarComponents which computes 
semi-algebraic descriptions of all connected semi-algebraic branch
curves of $\curve_{2, \bmA_1}\cap \R^2 \setminus \app(\curve_{2,
\bmA_1})$ (or $\curve_{2, \bmA_1}\cap \R^2$) using 
$\softO{\delta^6(\tau+\delta)}$ bit operations, where $\left( \delta,
\tau\right)$ is the magnitude of $\curveparam_1$ (hence $\delta$ is
the degree of $\curve$); see \cref{thm:saisotopygraph-plane}. Note
that this is $\delta$ times more than the complexity of the algorithm
in \cite{IslamPoteauxPrebet2023} for counting the connected components
of $\curve_\R$.

These semi-algebraic descriptions can be lifted in $\R^n$. 
Thanks to \cite{IslamPoteauxPrebet2023}, one can gather those 
semi-algebraic curves lying in the same connected components of
$\curve_\R$. Still, we are missing the "failure" points in 
$F_1 = \curve^{\bmA_1} \setminus \pi_2^{-1}\left( \app\left( \curve_{2,
\bmA_1} \right)  \right)$ (where $\pi_2$ denotes the projection on the
first two coordinates). 
To repair this, one could try to compute $F_1$.
However, a degree estimation shows that the degree of
$\app\left( \curve_{2, \bmA_1}\right)$ lies in $O\left( \delta^2
\right)$ and then $F_1$ has degree in $O(\delta^3)$. It seems unlikely to
compute $F_1$ without exceeding the previous cost. 

The basic idea on which we rely is the
follo\-wing. If, by chance, doing the same computation with a second
distinct matrix $\bmA_2$ yields failure points $F_2$ such that
$F_1^{\bmA_1^{-1}}\cap F_2^{\bmA_2^{-1}}=\emptyset$, then one can hope
to recombine the semi-algebraic descriptions without exceeding the
previous cost. This is what we do and we prove that all of this works
up to genericity assumptions on the
choice of $\bmA_1$ and $\bmA_2$. \\
We can now state our main result, expressed in terms of extrinsic quantities
(the magnitudes of the input polynomials). 

\begin{theorem}\label{thm:main}
  Let $\bmf=(f_1,\dots,f_{n-1})\subset \mathbb{Q}[x_1,\dots,x_n]$
  generating a radical ideal of dimension $1$, $\curve\subset \C^n$
  be the algebraic curve it defines and let $(d, h)$ be the maximum
  magnitude of the $f_i$'s.
  Let $0<\epsilon<1$ be a probability parameter.\\
  There exists a Zariski closed subset $\mathscr{Z}\subsetneq
  \GL_n(\C)$ such that for any $\bmA_1\in GL_n(\Q)
  \setminus\mathscr{Z}$, there exists a Zariski closed subset
  $\mathscr{Z}_{\bmA_1}$ such that for any $\bmA_2\in
  \mathscr{Z}_{\bmA_1}$, letting
  $\overline{\tau}$ be the maximum bit size of the entries of $\bmA_1$
  and $\bmA_2$, 
  \ComputeRealComponents$(\bmf,\varepsilon, \bmA_1, \bmA_2)$ 
  computes a semi-algebraic description
  of each connected component of
  $\mathcal{C}_{\mathbb{R}}=\curve\cap\R^n$ using
  $\simplifiedtotalcompspace$
  bit operations,
  with probability of success at least $1-\varepsilon$.
\end{theorem}

A few remarks are in order. First, note that when $n=3$, we obtain a
complexity which is
$\softO{d^{14}(h+\overline{\tau}+1)+\log(\frac{1}{\epsilon})d^{13}}$.
This shows that, currently, one can compute the connected components
of a curve in $\R^3$ (under our assumptions and $\overline{\tau}\in
\log\left( d^{O(1)} \right) $) faster than computing a graph isotopic
to the curve under study. \\
Also, the algorithm boils down to compute one dimensional
paramet\-rizations of the curve $\curve$ 
(after applying $\bmA_1$ and $\bmA_2$),
plus some extra points in $\curve$ (see \cref{sec:spacecurves}) and
then applies twice \GenPlanarComponents whose complexity actually
depends on the actual degree of the curve (which we bounded by
$d^{n-1}$). Hence, a more accurate complexity statement could be
stated. In particular, assuming that these parametrizations are given
(as in \cite{IslamPoteauxPrebet2023}), we obtain a complexity which is
"only" $\delta$ times more the one of \cite{IslamPoteauxPrebet2023}. 

Finally, note that if $\mathfrak{D}$ bounds the sum of the degrees of 
$\mathscr{Z}$ and
$\mathscr{Z}_{\bmA_1}$ then it would suffice to pick $\bmA_1$ and
$\bmA_2$ with entries of bit size in $\softO{\log(\frac{1}{\epsilon})
+ \log(\mathfrak{D})}$ to ensure a probability of success $>
1-\epsilon$. If $\mathfrak{D}\in d^{n^{O(1)}}$, only powers of $n$
would be affected in our complexity
statement  The proof
of \cref{sec:twoparams} suggests that $\mathcal{Z}_{\bmA_1}$ has
degree in $d^{O(n)}$. The set $\mathscr{Z}$ is the same as the one
considered in \cite{IslamPoteauxPrebet2023}. We conjecture that it
lies in $d^{O(n)}$ also. 

\paragraph*{Structure of the paper.}\cref{sec:planecurves} is devoted
to the design of \GenPlanarComponents. 
\cref{sec:spacecurves} contains our main algorithm
\ComputeRealComponents. \cref{sec:twoparams}
proves our genericity statement.

\section[\NoCaseChange{Plane Curves}]{\texorpdfstring{\NoCaseChange{Plane Curves}}{Plane Curves}}\label{sec:planecurves}

Let $f\in\Q[x_1,x_2]$ be a \emph{square-free} polynomial defining the curve $\curve\subset \C^2$ and let $\curve_\R =
\curve\cap \R^2$. 
Our approach slightly modifies the algorithm
\textsc{TopoNT}~\cite[Sec.~3.1]{MehlhornSagraloffWang2013}, which
itself builds upon~\cite{BerberichEmeliyanenkoKobelSagraloff2013}, and
produces a planar straight-line graph in $\mathbb{R}^2$,
isotopic to $\mathcal{C}_{\mathbb{R}}$.\\
\noindent
Let $\mathcal{G}$ be a graph embedded in $\R^2$, with vertex set
$\mathcal{V}(\mathcal{G}) \subset \R^2$ and edge set $\mathcal{E}(\mathcal{G})$.
We denote by $\mathcal{C}_{\mathcal{G}}$ the planar curve obtained by taking the
union of all vertices in $\mathcal{V}(\mathcal{G})$ together with the straight-line
segments connecting the endpoints of the edges of $\mathcal{G}$.\\
\noindent
An isotopy of $\R^n$ is a continuous map 
$\mathcal{I}$ from $\R^n \times [0,1]$ to $\R^n$
with $p\mapsto \mathcal{I}(p,0)$ being the identity map and $p\mapsto \mathcal{I}
(p,t)$ being a homeomorphism ($t\in [0,1]$). Two subsets $Y$ and $Z$ of $\R^n$ are isotopy
equivalent if there is an isotopy $\mathcal{I}$ of $\R^n$ such that $\mathcal{I}(Y,1)=Z$.\\
\noindent
We now introduce a refined graph structure, still isotopic to
$\mathcal{C}_{\mathbb{R}}$, but satisfying additional properties, which will serve 
as an auxiliary structure to obtain a semi-algebraic description of the connected components of $\calC_{\R}$.
Let $\calP \subset \operatorname{reg}(\calC)$ be a finite set of
points on $\curve$. We denote by $\scrK(\pi_1, \curve)\subset \C^2$
the solution set to $f = \frac{\partial f}{\partial x_2} =0$. 

\begin{definition} \label{saisotopygraph}
Let $\mathcal{H}=(\mathcal{V},\mathcal{E})$ be a graph, with $\mathcal{V}\subset \mathbb{R}^2$.
We say that $\mathcal{H}$ is a \textbf{semi-algebraic real isotopy graph} of $(\mathcal{C},\mathcal{P})$ if 
\begin{itemize}[leftmargin=*]
  \item $\mathcal{C}_{\mathbb{R}}$ is isotopy equivalent to
    $\mathcal{C}_{\mathcal{H}}$;
  \item the points of $\scrK(\pi_1,\mathcal{C}_{\mathbb{R}} )\cup
    \calP$ are embedded in $\mathcal{V}$;
  \item no two points of $\scrK(\pi_1,\mathcal{C}_{\mathbb{R}})$ have
    adjacent vertex in $\mathcal{H}$;
\item for each vertex $v\in\mathcal{V}$ corresponding to a point
  $\xi\in \calC_{\R}$; we are given a semi-algebraic description
  $\sigma_v$ of $\xi$;
\item each edge $e\in\mathcal{E}$ comes with a semi-algebraic
  description $\sigma_e$ of the real branch of 
$\calC_{\R}$ isotopic to $e$.
\end{itemize}
\end{definition}

\noindent
\textit{We assume that $f$ is in generic
coordinates in the sense of
\cite{BerberichEmeliyanenkoKobelSagraloff2013}.}
We design an algorithm \GenPlanarComponents 
which takes as input:
\begin{itemize}[leftmargin=*]
  \item $f\in \Z[x_1, x_2]$ square-free of magnitude $\left( \delta, \tau
    \right)$ such that $V(f) = \curve$;
  \item a zero-dimensional parametrization $\queryparam$ such that
    $\calP = \Zeroes(\queryparam) \subset \calC$; 
  \item $q\in \Z[x_1]$ dividing the resultant of $f, \frac{\partial
    f}{\partial x_2}$
    with respect to $x_2$. 
\end{itemize}
Let $\mathcal{V}_{\scrK},\mathcal{V}_{\calP}\subset\mathcal{V}$ be the vertices
corresponding to $\scrK(\pi_1,\mathcal{C}_{\mathbb{R}})$ and $\calP$.
It outputs a semi-algebraic real isotopy graph
$\calH = (\calV, \calE)$ of $(\curve, \calP)$, together with the subsets
$\calV_{\calP}$ and $\calV_q \subset \calV_{\scrK}$ corresponding to the vertices
whose $x_1$-coordinates are roots of $q$.
Let $(\delta, \tau)$, $(\delta_\querypoints, \tau_\querypoints)$, 
$(\delta_q, \tau_q)$
be the magnitudes of $f, \queryparam$ and $q$.
\begin{theorem}\label{thm:saisotopygraph-plane}
$\GenPlanarComponents(f,\queryparam,q)$ 
computes a se\-mi-algebraic real isotopy graph $\calH$ of $(\mathcal C,\mathcal P)$
using
\[
\complgenplanarcomponents\]
bit operations. Moreover, $|\mathcal{V}(\calH)|\in O(\delta^4+\delta \delta_\calP)$ and 
$|\mathcal{E}(\calH)|\in O(\delta(\delta^3+\delta_\calP+1))$.

\end{theorem}

\subsection{Outline of the construction}


The regular (resp. singular) locus of $\curve$ is denoted
$\reg(\curve)$ (resp. $\sing(\curve)$). Since $f$ is square-free,
$\sing(\curve) = V\left( f,{\partial f}/{\partial x_1},
{\partial f}/{\partial x_2} \right)$.

Let $\calS_k = V\left( f,
\partial^k f /{\partial x_2^k} \right)$. 
\textit{We assume that $f$ satisfies assumption $\hypS$: $\calS_k$ is finite for all $1\leq
k \leq d_2 = \deg(f, x_2)$}. 
Let $\calS(f) = \cup_{k=1}^{d_2} \calS_k$.\\
We denote by
$\pi_1$ the projection $(\xi_1, \xi_2) \to \xi_1$. \\
\noindent
In the \textit{projection phase} of the algorithm in
\cite[Section~3]{MehlhornSagraloffWang2013}, $\scrK(\pi_1,\mathcal{C})$ is
projected onto the $x_1$-axis. As 
in~\cite{MehlhornSagraloffWang2013}, we denote these values by
$\alpha_1, \ldots, \alpha_m$, and $I_1, \ldots, I_m$ for their
isolating intervals.  In the same phase, arbitrary values $\beta_0,
\ldots, \beta_{m+1} \in \mathbb{Q}$ are also computed, with $\alpha_m
< \beta_{m+1}$ and $\beta_{i-1}< \alpha_{i}<\beta_i$ for all $i = 0,
\ldots, m$.

Denote by $\gamma_1, \ldots, \gamma_t$ the points 
in $\pi_1(\calS(f))$, 
and by $J_1, \ldots, J_t$ their isolating intervals.  The procedure
for computing these intervals is analogous to the one used to obtain
the~$\alpha_i$ in~\cite{MehlhornSagraloffWang2013}.\\
\noindent
In the \textit{lifting phase}, isolating intervals of the real roots
of $f(\alpha_i, x_2)$ and $f(\beta_i, x_2)$ are computed, together with
their corresponding multiplicities. 
Since $f$ is in generic coordinates, these polynomials have at most one multiple
root (see proof of \cite[Thm.~5]{BerberichEmeliyanenkoKobelSagraloff2013}). 
In our setting, we additionally compute the real roots of $f(\gamma_i,
x_2)$ for $1\leq i \leq t$.\\
\noindent
The final phase for obtaining a graph isotopic to the curve is the
\textit{connection step}. 
The vertices of the graph~$\mathcal{G}$ are given by 
\[ V(\mathcal{G})=\{(x,y)\in \mathbb{R}^2\ |\ \exists i \ \text{with}\ x=\alpha_i\ \text{or}\ x=\beta_i\, \text{and}\ f(x,y)=0\}.\] 
At this point, connecting the vertices is straightforward;
see~\cite[Sec. 3.2.3]{BerberichEmeliyanenkoKobelSagraloff2013} for the details. 
Starting from this graph, we construct our graph $\mathcal{H}$. 
The vertices of this new graph are given by 
\[
V(\mathcal{H}) \;=\; V(\mathcal{G}) \,\cup\, \{(x,y)\in \mathbb{R}^2\ |\ \exists i \ \text{with}\ x=\gamma_i\ \text{and}\ f(x,y)=0\},
\]
while the edges of the new graph are obtained as follows. From each
vertex in $V(\mathcal{G})$, the same number of edges originate as in
$\mathcal{G}$. 
Suppose that 
$r$ points $(\gamma_{i_j}, y_j)$, for $j=1,\ldots,r$, of the set
\[
\{(x,y)\in \mathbb{R}^2 \mid \exists\, i \ \text{such that}\
x=\gamma_i \ \text{and}\ f(x,y)=0\} \setminus V(\mathcal{G})
\]
lie on the same edge $(v,w)$ of $\mathcal{G}$ (in the isotopic sense), 
and assume that the corresponding $\gamma_{i_j}$ are ordered increasingly. 
Then this edge is subdivided into $r+1$ edges, whose ordered endpoints are 
$v,\; (\gamma_{i_1},y_1),\; \ldots,\; (\gamma_{i_r},y_r),\; w $.
We refer to this construction as the procedure \textsc{UpdateEdges}, which
will be used in Algorithm~\ref{alg:refine-graph-H} to construct the edge set of $\mathcal{H}$.
We now explain how to determine to which edge of $\mathcal{G}$ the new 
vertices belong. Let $\alpha_k$ and $\alpha_{k+1}$ be two consecutive 
$x_1$-coordinates obtained as described above. For any choice of 
$\overline{x} \in (\alpha_k, \alpha_{k+1})$, the number of real roots of 
$f(\overline{x},y)$ is invariant, and it coincides with the number of edges 
of the graph $\mathcal{G}$ lying above any interval of the $x_1$-axis contained in 
$(\alpha_k, \alpha_{k+1})$ whose endpoints are consecutive $x_1$-coordinates 
of vertices of $\mathcal{G}$.
Therefore, if $\gamma_i \in (\alpha_k, \alpha_{k+1})$, the number of real 
roots of $f(\gamma_i,x_2)$ equals the number of intersections 
between $\mathcal{G}$ and the vertical line $\{x_1=\gamma_i\}$. By ordering 
the real roots of $f(\gamma_i,x_2)$, we can then determine to which edge of 
$\mathcal{G}$ each of these points belong.

For $(x, y)\in \R^2$, define $\sigma(x, y)$ as the sequence of the
signs of ${\partial^k f}/{\partial x_2^k}(x, y)$ for $k\in \{1,
\ldots, d_2\}$ .

\begin{theorem}\label{constantsignpattern}
Let $f\in \mathbb{R}[x_1,x_2]$ satisfying $\hypS$. 
Let $\pi_1(\calS(f)) = \{\xi_1,\dots,\xi_{\ell}\}$,
$K_i := (\xi_i, \xi_{i+1}) \subset \mathbb{R}$ for $i=1,\dots,\ell-1$ and
denote by $C_{i,1}, \ldots, C_{i, m_i}$ the connected components of
$\pi_1^{-1}(K_i)$. \\
Then, for all $i=1,\dots,\ell-1$ and $j=1,\dots, m_i$, the sequence $\sigma(x,y)$ is 
constant along $C_{i,j}$. 
Moreover each connected component $C_{i,j}$ of $\pi_1^{-1}(K_i)$ is uniquely identified 
by its sign-sequence.
\end{theorem}

\begin{proof}
Fix $i \in \{1,\dots,\ell-1\}$ and $j \in \{1,\dots,m_i\}$ and a 
component $C_{i,j}$ of $\pi_1^{-1}(K_i)$. 
By the definition of $\calS(f)$ and its finiteness, 
for every $(x,y)\in C_{i,j}$ and for all $k$ with $0<k< d_2$, the
derivative ${\partial^k f}/{\partial x_2^k}(x,y)$ does not vanish
along this semi-algebraic subset. 

Since $f$ and all its partial derivatives are
polynomials, they are continuous on $\mathbb{R}^2$.  Consequently,
their signs cannot change along a connected subset where they do not
vanish. 
Hence, the sign sequence $\sigma(x,y)$ remains constant for all
$(x, y) \in C_{i,j}$.

For the second statement, we appeal to the notion of \emph{Thom
encoding} (see \cite[Prop. 2.28]{BPR}). 
Let $\alpha \in K_i$ and consider 
$f(\alpha,x_2) \in \mathbb{R}[x_2]$. 
Its real roots correspond bijectively to the connected components of 
\(\pi_1^{-1}(K_i)\).
Indeed, for each real root $\beta$ of $f(\alpha,y)$, the point 
$(\alpha,\beta)$ lies on exactly one branch $C_{i,j}$, and every branch 
$C_{i,j}$ contains exactly one such point for the chosen $\alpha$.
Each root of $f(\alpha,x_2)$ admits a unique Thom encoding, that is, 
a unique sequence of signs of the derivatives $\frac{\partial^k
f}{\partial x_2^k}$ for $1 \leq k \leq d_2$. 
This sequence uniquely identifies the corresponding component
$C_{i,j}$ that contains the root. 
\end{proof}
\noindent
We denote by $\sigma_{i,j}$ the sign condition associated with the sequence of polynomials
$\left( {\partial^k f}/{\partial x_2^k}(x,y),{1 \le k \le \deg_y(f)} \right)$
along $C_{i,j}$.

\begin{corollary}
Using the above notation, along each arc of the curve $\mathcal{C}$
corresponding to an edge of the graph $\mathcal{H}$, the sign of the sequence of 
partial derivatives of ${\partial^k f}/{\partial x_2^k}$ for
$1\leq k \leq d_2$ remains constant.
Moreover, each edge of $\mathcal{H}$ whose endpoints have the same 
$x_1$-coordinates corresponds to a distinct sign pattern.
\end{corollary}
\begin{proof}
Observe that $\pi_1(\calS(f))$ is the union 
of the points $\alpha_i$ and $\gamma_i$ defined above. 
For the vertices of the graph $\mathcal{H}$ we only have additional
points corresponding to those whose 
$x_1$-coordinates are $\beta_i$. 
Then every curve arc that is isotopic to an 
edge of $\mathcal{H}$ 
is contained in one of the component $C_{i,j}$. 
This proves the first part.
For the second one, we use Thom encodings as in the previous 
proof.
\end{proof}
Recall that for each $i\in \{1,\dots,\ell-1\}$ and 
$j \in \{1,\dots,m_i\}$, the sign pattern
$\sigma_{i,j}$ is an element of $\{-1,1\}^{d_2}$, 
where $d_2 = \deg_{x_2}(f)$.  

We define the semi-algebraic set 
$\operatorname{Reali}(\sigma_{i,j}) \subset \mathbb{R}^2$ 
\[
\operatorname{Reali}(\sigma_{i,j})
:= \left\{ (x,y) \in \mathbb{R}^2 \ \middle|\ 
\left( \operatorname{sign}\!\left( 
       {\partial^k f}/{\partial x_2^k}(x,y) 
       \right) \right)_{1 \le k \le \delta} 
= \sigma_{i,j} \right\}.
\]
\noindent
The sets $\operatorname{Reali}(\sigma_{i,j})$
are given by a semi-algebraic description fixing the signs of 
${\partial^k f}/{\partial x_2^k}$ for $1 \le k \le d_2$. We
denote it by $\Sigma_{i,j}$.\\
\noindent
Let us now consider an edge of the graph $\mathcal{H}$, and denote its endpoints by 
$v_1$ and $v_2$. If both vertices belong to $\calS(f)$, then this edge is isotopic 
to one of the components $C_{i,j}$ described in Theorem~\ref{constantsignpattern}, 
and its semi-algebraic description is given by
\begin{equation}\label{eq:sas1}
\{ \pi_1(v_1) < x_1 < \pi_1(v_2) \} \wedge
\Sigma_{i,j}\wedge  \{ f = 0 \} .
\end{equation}
\noindent
If instead one endpoint lies in $\{(x,y)\in \mathbb{R}^2\ | \ \exists i \ \text{such that}\ 
x=\beta_i\ \text{and}\ f(x,y)=0\}$, then we must first 
identify the component $C_{i,j}$ to which the edge lies, and the corresponding 
semi-algebraic description is then given by \eqref{eq:sas1}.
To determine this component $C_{i,j}$, 
observe that from every vertex in~$\{(x,y)\in \mathbb{R}^2\ | 
\ \exists i \ \text{such that}\ x=\beta_i\ \text{and}\ f(x,y)=0\}$ exactly two edges originate, 
and that two vertices in this set are never adjacent (i.e., they are never 
the endpoints of the same edge). 
Thus, if for instance $v_1$ lies in this set, it suffices to look at the other vertex 
to which it is connected, different from $v_2$, in order to identify the 
corresponding component $C_{i,j}$.

\noindent
Thus, the graph $\mathcal{H}$ is a semi-algebraic isotopy graph in the sense of
Definition~\ref{saisotopygraph}, where the information $\sigma_e$ associated with each
edge $e$ corresponds to some $\sigma_{i,j}$, in such a way that the edge $e$
either coincides with or is contained in the component $C_{i,j}$.
Each connected component $C_i$ of
$\mathcal{C}_\R$ is semi-algebraically described by taking the union of the semi-algebraic descriptions
of all edges belonging to $C_i$, together with the vertices that lie
on it.

Vertices of $\mathcal{H}$ are described by isolating intervals,
together with polynomial equations.
Let
$w_k$
be 
$\operatorname{Res}_{x_2}\!(f,{\partial^k f}/{\partial
x_2^k})$,
where $\operatorname{Res}_{x_2}$ denotes the resultant with
respect to the variable $x_2$.
For $k \ge 2$ set
$g_k := \gcd(w_1,\dots,w_{k-1},w_k)$
and 
$\overline{w}_k := \frac{w_k}{g_k}$.
Vertices lying above
$\pi(\mathcal{K}(\pi_1,\mathcal{C}))$
are described by 
$w_1(x_1)=0$ and $f(x_1,x_2)=0$,
while all remaining vertices correspond to the equations
$\overline{w}_k(x_1)=0$ and $f(x_1,x_2)=0$,
for an appropriate value of $k$.

\noindent
To incorporate $\calP$ into $\mathcal{H}$, let $(\lambda,\theta_2)$ be a
zero-dimensional parametri\-zation of $\calP$ and set
$\overline{\lambda} := \lambda / \gcd(w_1\cdots w_\delta,\lambda)$.
The solutions of $\overline{\lambda}=f=0$ define new regular vertices,
which are added to $\mathcal{H}$ and connected using
\textsc{UpdateEdges}; the subset $\mathcal{V}_{\calP}$ is then identified
using the parametrization of $\calP$.


\noindent
Finally, the subset $\calV_q$ is obtained by identifying, among the
vertices in $\calV_{\scrK}$, those whose $x_1$-coordinate is a root of
$q$.

\subsection{Strategy for Assigning a Sign Pattern}
On each such smooth connected branch of
$\mathcal C_{\mathbb R} \setminus \bigl(\scrK(\pi_1,\calC)\cup \calS_k\bigr)$,
the sign of $\partial^k f/\partial x_2^k$ is constant.
Therefore, it suffices to evaluate $\partial^k f/\partial x_2^k$ at a
single point on each such branch in order to determine the sign pattern
along the entire branch.

We denote by $\mathcal{Q}_k$ is a finite subset meeting all connected
components of $\mathcal C_{\mathbb R} \setminus
\bigl(\scrK(\pi_1,\calC)\cup \calS_k\bigr)$.
The procedure \textsc{PropagateSigns} takes as input
a graph $\mathcal H$ satisfying
$\mathcal{K}(\pi_1,\mathcal{C}_{\mathbb R}) \cup \mathcal{S}(f)
\;\subset\;
V(\mathcal H)$,
together with a family of sign maps
$\{\, s_k : \mathcal{Q}_k \to \{-1,0,+1\} \,\}_{k=1}^{d_2}$,
where $s_k(q) = \operatorname{sign}\!({\partial^k f}/{\partial
x_2^k}(q))$. It 
returns, for each edge $e\in\mathcal{E}(\mathcal H)$, a sign vector
$\sigma_e = \bigl(\sigma_{e,1},\dots,\sigma_{e,d_2}\bigr)
\in \{-1,0,+1\}^{d_2}$
such that $\sigma_{e,k}$ coincides with the sign of
$\partial^k f/\partial x_2^k$ at any point of the smooth connected branch of
$\mathcal{C}_{\mathbb R}$ encoded by $e$.

\begin{lemma}\label{lem:computesignsamples}
Let $f\in\mathbb{Q}[x_1,x_2]$ be of magnitude $(\delta,\tau)$, in
generic coordinates and satisfying $\hypS$. Let
$w_1,\dots,w_\delta\in\mathbb{Q}[x_1]$ be defined by
$w_k := \operatorname{Res}_{x_2}\!\bigl(f,\partial_{x_2}^{\,k} f\bigr)$ for $k=1,\dots,\delta$.
Algorithm~\textsc{SignSamples}, given as input $(f,w_1,\dots,w_\delta)$,
computes for each $k\in\{1,\dots,\delta\}$ a finite set $\mathcal{Q}_k$ as above,
together with the corresponding sign map $s_k$,
using
$\widetilde{O}(\delta^7 + \delta^6\tau)$
bit operations.
\end{lemma}
\begin{proof}[Proof of correctness]
For each $k\in\{1,\dots,d_2\}$, the algorithm constructs the set
$\mathcal{Q}_k$ as follows.
We define $P_k := w_1 w_k$, with the convention that
$P_1 := w_1$, and denote by
$\zeta_{k,1} < \cdots < \zeta_{k,\ell_k}$
the real roots of $P_k$.
We need to construct rational values
$\theta^{(k)}_0 < \zeta_{k,1} < \theta^{(k)}_1 < \cdots <
\zeta_{k,\ell_k} < \theta^{(k)}_{\ell_k}$,
and to include in $\mathcal{Q}_k$ the real fibers of
$\pi_1|_{\mathcal{C}_{\mathbb R}}$ above these values.
By construction, each open interval $(\zeta_{k,i},\zeta_{k,i+1})$
contains no real root of $P_k$.
Since the derivative $P_k'$ has at least one real root in each such
interval, sufficiently accurate approximations of the real roots of $P_k'$ provide values
$\theta^{(k)}_i$ in the intervals $(\zeta_{k,i},\zeta_{k,i+1})$.
As $w_k(x_1)=Res_{x_2}(f,\partial_{x_2}^k f)$, this implies that
$\partial^k f/ \partial x^k_2$ does not vanish at any real point of the fiber
$\pi_1^{-1}(\theta^{(k)}_i)\cap\mathcal{C}_{\mathbb R}$
 and the sign of $\partial^k f/\partial x^k_2$ is well defined at every
point of $\mathcal{Q}_k$.
For $\theta^{(k)}_0$ and $\theta^{(k)}_m$, we use a bound on the
absolute values of the roots of $w_k$ for $k=1,\dots ,d_2$
(see~\cite[Lem.~10.2]{BPR}) and choose two rational numbers
outside this bound accordingly.
\end{proof}

Proof of complexity is postponed in \cref{sec:planecurvealgo} and an
explicit description of the algorithm is given in the appendix.
\noindent
We define the algorithm \textsc{AssignSigns} which
takes as input an integer $k\ge 1$ and a polynomial
$g\in\mathbb{Q}[x_1]$.
It outputs a finite set $\mathcal{Q}_k\subset \mathcal{C}_{\mathbb{R}}$
together with a sign map
$s_k:\mathcal{Q}_k\to\{-1,0,1\}$.
The set $\mathcal{Q}_k$ is obtained by computing isolating intervals for
the real solutions of the system
$g(x_1)=f(x_1,x_2)=0$,
and for each $q\in\mathcal{Q}_k$ the value $s_k(q)$ is defined as
$s_k(q):=\operatorname{sign}\bigl(\partial_{x_2}^k f(q)\bigr)$.

\subsection{Algorithms and Complexity Analysis} \label{sec:planecurvealgo}
\noindent
\textsc{IsolatingRealSol} takes as input a univariate polynomial
$R(x_1)\in\mathbb{Q}[x_1]$ and a bivariate polynomial
$F(x_1,x_2)\in\mathbb{Q}[x_1,x_2]$ defining a zero-dimensional system, and
returns isolating boxes for the real solutions of
$R(x_1)=F(x_1,x_2)=0$.

\begin{algorithm}[!htb]
\caption{\PlanarComponents}
\label{alg:refine-graph-H}
\begin{algorithmic}[1]
\Require $f$ in $\mathbb{Q}[x_1,x_2]$, which satisfies
$\calS$, and  
a zero-dimensional parametrization $(\lambda(x_1),\theta_2)$ defining
$\calP \subset \reg(\curve)$,
both in generic coordinates, and 
$q\in \Z[x_1]$ such that $q$ divides $w_1$.

\Ensure A semi-algebraic isotopy graph $\mathcal{H}$ of $(\calC_\R,\calP)$ together
with a sign vector $\sigma_e$ for each edge $e\in\mathcal{E}(\mathcal{H})$.

\State $\mathcal{G} \;\leftarrow\; \textsc{TopoNT}(f)$
\Comment{$\mathcal{G}$ is isotopic to $\curve_{\mathbb{R}}$, see \cite{MehlhornSagraloffWang2013}}

\State $V(\mathcal{H}) \;\leftarrow\; V(\mathcal{G})$

\State $V_q\; \leftarrow \; \textsc{IsolatingRealSol}(q(x_1),0)$

\State $w_1 \;\leftarrow\; \operatorname{Res}_{x_2}\!\bigl(f,\partial_{x_2} f\bigr)$

\For{$k = 2,\ldots,\delta$}
  \State $w_k \;\leftarrow\; \operatorname{Res}_{x_2}\!\bigl(f,\partial_{x_2}^{\,k} f\bigr)$
  \State $g_k \;\leftarrow\; \gcd(w_1\cdots w_{k-1},w_k)$
  \State $\overline{w}_k \;\leftarrow\; w_k/g_k$
  \State $V(\mathcal{H}) \leftarrow 
  V(\mathcal{H})\cup \textsc{IsolatingRealSol}(\overline{w}_k(x_1),\ f(x_1,x_2))$
\EndFor 
\State $\overline{\lambda}\leftarrow \lambda/gcd(w_1\cdots w_\delta,\lambda)$
\State $\calV(\calH)\leftarrow \calV(\calH)\cup \textsc{IsolatingRealSol}(\overline{\lambda}(x_1),f(x_1,x_2))$
\State $V_\calP\; \leftarrow\; \textsc{IsolatingRealSol}(\lambda(x_1),(\lambda' x_2-\theta_2)(x_1,x_2))$
\State $\calH \;\leftarrow\; \textsc{UpdateEdges}(\mathcal{V}(\calH),\mathcal{G})$

\State $\{(\mathcal{Q}_k,s_k)\}_{k=1}^d \;\leftarrow\;
\textsc{SignSamples}\ (f,w_1,\dots,w_d)$

\State $\{\sigma_e\}_{e\in\mathcal{E}(\mathcal{H})}
\;\leftarrow\; \textsc{PropagateSigns}\ \!\bigl(\mathcal{H},\{s_k\}_{k=1}^d\bigr)$

\State \textbf{return} $\bigl(\mathcal{H},\{\sigma_e\}_{e\in\mathcal{E}(\mathcal{H})}\bigr)$

\State \textbf{return} $\mathcal{H}$
\end{algorithmic}
\end{algorithm}

\noindent
\begin{proposition}
  Algorithm \PlanarComponents is correct.  
\end{proposition}

\begin{proof}
Algorithm~$\textsc{TopoNT}(f)$ returns a graph isotopic to
$\mathcal{C}_{\mathbb R}$ \cite[Thm.~6]{MehlhornSagraloffWang2013}.
Step~3 identifies the vertices corresponding to $q$, while the loop over
$k$ computes the additional vertices associated with $\mathcal{S}(f)$.
Steps~11--13 introduce the vertices corresponding to $\calP$.
The procedure \textsc{UpdateEdges} refines the graph without changing its
isotopy class.
Finally, the correctness of the sign information computed by
\textsc{SignSamples} and propagated by \textsc{PropagateSigns} follows
from Lemma~\ref{lem:computesignsamples} (it applies since $f$ is
in generic coordinates).
\end{proof}
\noindent
Let $(\delta,\tau)$ denote the magnitude of $f$ and
$(\delta_{\mathcal{P}},\tau_{\mathcal{P}})$ the magnitude of $\calP$.
We can assume that the polynomials
$w_1,\dots,w_\delta$ are pairwise coprime.
For $P \in \mathbb{Z}[X]$ of magnitude $\left(p, \eta \right)$, 
the logarithm of the separation distance of the roots of $P$ lies in 
$\log_2\!\bigl(1/\operatorname{sep}(P)\bigr) = O(p\eta + p\log p)$ (see~\cite[Cor.~10.22]{BPR}).
Moreover, by~\cite[Lem.~66]{MelczerSalvy2021}, isolating the real
roots bit precision $\overline{k}$ of such a polynomial $P$ 
can be carried out within
$\tilde{O}\!\left(p^3 + p^2\eta + p\overline{k}\right)$
bit operations. 


\begin{lemma}\label{lem:wij}
  The bit precision required for separating the real roots of
  $\omega_i$ and $\omega_j$ is in
  $\widetilde{O}(\delta^3(\tau+\delta\log \delta))$.
\end{lemma}
\begin{proof}
The coefficients of $\partial^k f/\partial y^k$ have bit
size boun\-ded by
$\tau_k := \tau + O(k\log \delta)$. 
Applying~\cite[Lem.~65]{MelczerSalvy2021} with degree bound $\delta$ and coefficient bit 
size bound
$\max\{\tau,\tau_k\}=\tau_k$, we obtain that $w_k$ has magnitude $(D,\eta_k)$ with
$D = O(\delta(\delta-k))$ and
$\eta_k = O\!\bigl(\delta(\log \delta + \tau_k)\bigr)
        = O\!\bigl(\delta(k\log \delta+\tau)\bigr)$.
The polynomial $w_i w_j$ has magnitude
$(O(\delta(2\delta-i-j)),O(\delta(2\tau+(i+j)\log\delta)))$; applying the separation
bound yields a required precision in
$\widetilde{O}(\delta^3(\tau+\delta\log\delta))$ bits.
\end{proof}

\begin{lemma}\label{lem:wilambda}
  The bit precision required for separating the roots of $\omega_i$ and $\lambda$ is in
  $\widetilde{O}((\delta^2+\delta_\calP)(\delta\log \delta+\tau+\tau_\calP))$.
\end{lemma}
\begin{proof}
We proceed as above. The product $w_i\lambda$ has magnitude
$(O(\delta(\delta-i)+\delta_\calP), O(\tau+i\log \delta+\tau_\calP))$. 
The bits of precision required to
isolate all its roots are $\widetilde{O}((\delta^2+\delta_\calP)(\delta\log \delta+\
\tau+\tau_\calP))$.
\end{proof}

\begin{proof}[Proof of complexity Lemma \ref{lem:computesignsamples}]
Let $P_i := w_1 w_i$, whose magnitude, as well as that of $P'_i$,
is bounded by $(2\delta^2,\,\delta(\tau+\delta\log\delta))$.
Isolating the roots of $P'_i$ with sufficient precision to separate them from the
roots of $w_k$, for $2\le k\le \delta$, costs
$\widetilde{O}(\delta^6+\delta^5\tau)$ bit operations.


\noindent
By \cite[Prop.~20]{DiattaDiattaRouillierRoySagraloff2022}, it is possible
to compute isolating intervals for all solutions of the system
$P'_1(x_1) = f(x_1,x_2) = 0$
together with the sign of $\partial^k f/\partial x_2^k$ at these points,
using $\widetilde{O}(\delta^6 + \delta^5\tau)$ bit operations.
Since we need to solve this system for each $k\in\{1,\dots,\delta\}$ ,
the total cost of this step is
$\widetilde{O}(\delta^7 + \delta^6\tau)$ bit operations.

\noindent
Finally, we handle the two points $\theta_0$ and $\theta_m$.
Their bit size is bounded by $\widetilde{O}(\delta(\tau+\delta))$
\cite[Lem.~10.2]{BPR}.
Consider the systems
$x_1 - \theta_j = f(x_1,x_2) = 0$,
for $j\in\{0,m\}$.
By \cite[Prop.~20]{DiattaDiattaRouillierRoySagraloff2022}, isolating
their real solutions and determining the sign of
$\partial^k f/\partial x_2^k$ at these points can be performed using
$\widetilde{O}(\delta^6 + \delta^5\tau)$ bit operations.
\end{proof}

\begin{proposition}\label{complexityplanarcomponents}
 Algorithm~\PlanarComponents uses
 \begin{equation*}
   \widetilde{O}((\delta^6+\delta^3\delta_\querypoints+\delta
   \delta_{\querypoints}^2)
   (\tau+\delta)+(\delta^5+\delta^3\delta_\querypoints+\delta_\querypoints^2)
   \tau_\querypoints+\delta^5\delta_\querypoints+\delta_\querypoints^3\\+\delta_q^2\tau_q)
 \end{equation*}
bit operations. Moreover, the output graph $\calH$ satisfies
  $|\calV(\calH)| \in O(\delta^4 + \delta\,\delta_{\calP})$ and 
  $|\mathcal{E}(\calH)| \in O\bigl(\delta(\delta^3 + \delta_{\calP} + 1)\bigr)$.

\end{proposition}

\begin{proof}
  Computing all partial derivatives $\partial_y^k f$ requires no more
than $\widetilde{O}\bigl(\delta^3(\tau+\delta\log\delta)\bigr)$ bit
operations.

The cost of Algorithm~\textsc{TopoNT}, as given in~\cite{MehlhornSagraloffWang2013},
is $\widetilde{O}(\delta^6 + \delta^5\tau)$. 
To identify the subset $\mathcal{V}_q \subset \mathcal{V}_{\mathcal K}$, it is enough
to isolate the roots of $q$ with the separation bound induced by the roots of
$w_1$, which costs $\widetilde{O}(\delta_q^3+\delta_q^2\tau+\delta_q\delta^3\tau)$ bit operations.

\noindent
Following the proof of~\cite[Corollary~2]{DiattaDiattaRouillierRoySagraloff2022},
the bit complexity of computing the bivariate resultant
$w_k$ lies in 
$\tilde{O}\!\bigl(\delta^4 \tau_k + \delta^5\bigr)=\tilde{O}\!\bigl(\delta^4(\tau+k
\log(\delta))+\delta^5\bigr)$.
We can therefore compute $\overline{\lambda}$, a divisor of $\lambda$
coprime with all $w_i$, at negligible cost, and with magnitude bounded by
$(\delta_{\mathcal P},\, O(\tau_{\mathcal P} + \delta_{\mathcal P}))$,
see~\cite{Mahler62}.
Since the precision bound of Lemma~\ref{lem:wij} is independent of $i$ and $j$,
the roots of all $w_k$ can be isolated with this precision so as to be mutually
separated. Accounting also for the separation from the roots of $\overline{\lambda}$
(Lemma~\ref{lem:wilambda}), the required precision is dominated by
$\widetilde{O}\!\Bigl(
(\delta^3+\delta_{\calP})(\tau+\delta)
+(\delta^2+\de_\calP)(\tau_{\calP}+\de_\calP)
\Bigr)$.

\noindent
The real roots of $\omega_k$ with this precision can be computed using
$\widetilde{O}(\delta^6+(\delta^5+\delta^2\delta_\calP)\tau+
\delta^4(\tau_\calP+\delta_\calP)+\delta^2\delta_\calP\tau_\calP)$ bit operations, 
while for the roots of $\lambda$ with the bit precision of \cref{lem:wilambda} we need 
$\widetilde{O}(\delta_\calP^3+\delta_\calP^2\tau_\calP+(\delta_\calP(\de^2+\de_\calP)
(\delta+\tau+\de_\calP+\tau_\calP))$
bit operations. Since $k\leq \delta$ the overall bit complexity for computing all isolationg intervals 
for $x_1$-coordinates of the vertices of $\mathcal{H}$ is bounded by 
\[
\widetilde{O}((\delta^6+\delta^3\delta_\calP+\delta_\calP^2)
(\tau+\delta)+(\delta^5+\delta^3\delta_\calP+\delta_\calP^2)\tau_\calP+
\de^5\de_\calP+\delta^2\de_\calP^2+\delta_\calP^3).
\]
To compute isolating boxes for all vertices of $\mathcal H$ lying above the roots
of $w_i$, for $i=2,\dots,\delta$, we apply
\cite[Prop.~20]{DiattaDiattaRouillierRoySagraloff2022}.
Isolating the real solutions of
$w_i(x_1)=f(x_1,x_2)\partial^k_{x_2}f(x_1,x_2)=0$
and identifying the common roots of $f$ and $\partial^k_{x_2}f$
costs $\widetilde{O}(\delta^6+\delta^5\tau)$ bit operations.
Repeating this for $i=2,\dots,\delta$ yields a total cost
$\widetilde{O}(\delta^7+\delta^6\tau)$.

While isolating the points in $\mathcal P$, we also isolate the real solutions of
$\overline{\lambda}(x_1)=f(x_1,x_2)=0$ using
$\widetilde{O}(\dep^2\tp+\dep^3+\de^5\tau+\de^6
+\de(\de^2+\dep)(\dep\tau+\de\tp+\de\dep))$ bit operations.
To identify the isolating intervals corresponding to the points in $\mathcal P$,
we apply \cite[Thm.~47]{MelczerSalvy2021}, which computes isolating intervals
for the real points encoded by the zero-dimensional parametrization
$(\lambda(x_1),\theta_2)$ in
$\widetilde{O}(\dep^3+\dep^2\tp+\dep k)$ bit operations, where
$k=\de(\dep\tau+\de\tp+\de\dep)+\de^4+\de^3\tau$ follows from
\cite[Prop.~15]{DiattaDiattaRouillierRoySagraloff2022} applied to
$\lambda=f=0$ and $w_i=f=0$.

\noindent
Summing the complexity bounds of all steps, the refinement procedure that
computes the graph $\mathcal{H}$ from the isotopy graph $\mathcal{G}$ uses
$\widetilde{O}((\delta^6+\delta^3\delta_\querypoints+\de\dep^2)(\tau
+\delta)+(\delta^5+\delta^3\delta_\querypoints+\delta_\querypoints^2)
\tau_\querypoints+\de^5\de_\calP+\delta_\querypoints^3)$.

\noindent
Combining the above with the complexity of
\textsc{SignSamples}, the claimed complexity bound for
\PlanarComponents follows.
\end{proof}

\subsection{Final algorithm}

Let $f \in \Q[x_1, x_2]$ of magnitude $(\delta, \tau)$,
in generic coordinates, defining the curve $\curve \subset \C^n$.
We no more assume that $f$ satisfies $\calS$. 
Let $\calP \subset \curve$. 
If
$f$ does not satisfy $\calS$, then there exists $k$ such that
$g = \gcd\left( f, {\partial^k f} / {\partial x_2^k} \right)$ has
positive degree. Repeating this recursively on $g$
and $f / g$ one obtains a factorization of $f$ into polynomials
satisfying $\calS$. We name this algorithm \textsc{PreCond} (see the
appendix).

\begin{lemma}\label{lem:precond}
  Algorithm \textsc{PreCond} uses $\compprecond$ bit operations. It outputs
  polynomials $(f_i)_{1 \le  i \le r}$ of magnitude $\left( \delta_i,
  \tau_i\right)_{1 \le  i \le  r}$ such that $\sum_{i=1}^r \delta_i =
  \delta$ and $\sum_{i=1}^r \tau_i \in \softO{\tau + \delta}$. 
\end{lemma}
\begin{proof}
  The computations of \textsc{PreCond} can be organized into a binary
  tree whose nodes are labeled by the polynomials given as input to
  \textsc{PreCond}. For each node $\nu$, we denote by $(\delta_\nu,
  \tau_\nu)$ the magnitude of its labeled polynomial. The call to
  \textsc{PreCond} at node $\nu$ performs at most $\delta$ gcd
  computations are performed, each of them with a bit cost in
  $\softO{\delta_{\nu}^4\tau_{\nu} + \delta_{\nu}^5}$.  
  Let $\nu$ be a node with children
  $\nu'$ and $\nu''$. The cumulated cost 
  of all descendants to $\nu$ 
  $\mathsf{C}(\delta_\nu, \tau_\nu)$ is such that 
  $\mathsf{C}(\delta_\nu, \tau_\nu) \in \softO{\delta_\nu^5\tau_\nu
  +\delta_\nu^6} + \mathsf{C}(\delta_{\nu'}, \tau_{\nu'})
  + \mathsf{C}(\delta_{\nu''}, \tau_{\nu''})$ with $\delta_{\nu'} +
  \delta_{\nu''} = \delta_\nu$ and, by \cite{Mahler62}, 
  $\tau_{\nu'} + \tau_{\nu''} \in \softO{\tau_\nu + \delta_\nu}$.
  Solving this recurrence yields the result. The final
  degree and bit size estimates are immediate from the fact that $f=
  f_1 \cdots f_r$ and \cite{Mahler62}. 
\end{proof}

We can compute a semi-algebraic real isotopy graph for each square-free factor
$f_i$ of $f$ using \PlanarComponents.
To obtain a single graph for the curve defined by $f$, we must
connect these partial graphs at the points of
$\mathcal{I}_{i,j} := V_{\mathbb{R}}(f_i,f_j)$
For $i\neq j$ in $\{1,\dots,r\}$, set 
$\mathcal{I}_i := \bigcup_{j\neq i}\mathcal{I}_{i,j}$.
We denote by \textsc{ConnectionPoints} the subroutine that takes as input
$f_1,\dots,f_r\in\mathbb{Q}[x_1,x_2]$ and returns zero-dimensional
parametrizations $\mathscr{I}_i$ for the sets $\mathcal{I}_i$ for $i=1,\dots,r$.
\begin{lemma}\label{lem:connectionpoints}
  \textsc{ConnectionPoints} uses $\softO{\delta^6+\delta^4\tau}$ bit operations. 
  It outputs zero-dimensional parametrizations $\mathscr{I}_i$
  of magnitude bounded by $(d_i\delta,\softO{\delta(\delta+\tau)})$.
\end{lemma}
\begin{proof}
Such a parametrization for the system
$f_i=f_j=0$ can be computed using $\softO{(\tau+\delta)(\delta_i\delta_j)^2+(\delta_i\delta_j)^3}$,
and its magnitude is bounded by 
$(\delta_i\delta_j,\softO{(\tau_j\delta_i+\tau_i\delta_j)+\delta_i\delta_j})$.
Since this computation has to be performed for all distinct pairs $(i,j)$,
we obtain the claimed complexity.
Fixing $i$ and taking the union of the corresponding zero-dimensional
parametrizations yields the set $\mathscr I_i$ with the 
claimed magnitude bound;
the cost of this union is negligible \cite{SaSc17}.
\end{proof}
After adding the points of $\mathcal{I}_i$ to the semi-algebraic isotopy graph
of $V_{\mathbb{R}}(f_i)$ for each $i\in\{1,\dots,r\}$, we can connect all partial
graphs and obtain a single graph isotopic to $\mathcal{C}_{\mathbb{R}}$.

We define a procedure \textsc{ConnectGraphs} that takes as input the
semi-algebraic isotopy graphs $\mathcal{H}_1,\dots,\mathcal{H}_r$ associated
with the real curves $V_{\mathbb{R}}(f_1),\dots,V_{\mathbb{R}}(f_r)$ and returns
a semi-algebraic isotopy graph of $\mathcal{C}_{\mathbb{R}}$.
The procedure works iteratively: set $\mathcal{H}\leftarrow\mathcal{H}_1$ and,
for $i=2,\dots,r$, merge $\mathcal{H}$ with $\mathcal{H}_i$ by identifying the
vertices corresponding to the same real intersection points in 
$\mathcal{I}_{i}$; the edge sets are then merged
accordingly.
\begin{algorithm}[!htb]
\caption{\GenPlanarComponents}
\label{alg:saig-general}
\begin{algorithmic}[1]
\Require $f\in\mathbb{Q}[x_1,x_2]$, and  
a zero-dimensional parametrization $(\lambda(x_1),\theta_2)$ defining
$\mathcal{P}\subset\calC$, both in
generic coordinates, and $q\in\mathbb{Z}[x_1]$ such that $q$ divides
$Res_{x_2}(f,\partial_{x_2}f)$.
\Ensure A semi-algebraic isotopy graph $\mathcal{H}$ of $(\mathcal{C}_{\mathbb{R}},\mathcal{P})$
together with a sign vector $\sigma_e$ for each edge $e\in\mathcal{E}(\mathcal{H})$.
\State $(f_1,\dots,f_r)\leftarrow \textsc{PreCond}(f)$
\State $(\mathscr{I}_i)_i \;\leftarrow\;
\textsc{ConnectionPoints}(f_1,\dots,f_r)$, 
\For{$i=1,\dots,r$}
  \State $(\mathcal{H}_i,\calV_{\calP_i} \cup \calV_{i}) \;\leftarrow\;
  \PlanarComponents\bigl(f_i,\mathcal{P}\cup\mathcal{I}_i,1\bigr)$
\EndFor
\State $\mathcal{H} \;\leftarrow\;
\textsc{ConnectGraphs}\bigl(\mathcal{H}_1,\dots,\mathcal{H}_r\bigr)$
\State $\calV_q\leftarrow \textsc{IsolatingRealSol}(q(x_1),0)$
\State $\calV_\calP\leftarrow \bigcup_i \calV_{\calP_i}$
\State \Return $(\calH,\calV_\calP,\calV_q)$
\end{algorithmic}
\end{algorithm}
\begin{lemma}\label{connectgraphs}
\textsc{ConnectGraphs} uses
$\widetilde{O}(\delta^4+\delta\delta_\querypoints)$ bit operations.
\end{lemma}
\begin{proof}
The procedure scans the vertex sets of $\mathcal{H}_1,\dots,\mathcal{H}_r$
and merges vertices corresponding to the same intersection points.
From the bound on the number of vertices in 
Proposition~\ref{complexityplanarcomponents},
we conclude that the number of bit operations is in  
$\widetilde{O}(\sum_i(\delta_i^4+\delta_i(\delta_\querypoints+
\de_i\delta))$.
\end{proof}

Observe that the vertex set $\calV_\scrK$ associated to the graph $\calH$ 
satisfies $V_{\mathscr K}
\;=\;
\bigcup_{i=1}^r \bigl( V_{\mathscr K_i} \cup V_i \bigr),
$
where $V_{\mathscr K_i}$ denotes the vertex set associated to $\scrK(\pi_1,\calC_i)$.
This fact is used to identify the subset
$V_q \subset \calV_{\mathscr K}$ once the real roots of $q$ have been isolated.

\begin{proof}[Proof of \cref{thm:saisotopygraph-plane}]
The complexities of \textsc{PreCond}, \textsc{ConnectionPoints}, and
\textsc{ConnectGraphs} follow from Lemma~\ref{lem:precond},
Lemma~\ref{lem:connectionpoints}, and Lemma~\ref{connectgraphs}, respectively.
We then apply the complexity bound of \textsc{PlanarComponents}
(Proposition~\ref{complexityplanarcomponents}) with the substitutions
$(\delta,\tau)\mapsto(\delta_i,\softO{\tau+\delta})$ and
$(\delta_\calP,\tau_\calP)\mapsto(\delta_\calP+\delta_i\delta,\,
\tau_\calP+\softO{\delta(\tau+\delta)})$ for each $i$,
and sum over all indices.
The treatment of the set $\calV_q$ follows the same strategy as in the proof of
Proposition~\ref{complexityplanarcomponents}, with the only additional cost
$\widetilde{O}(\delta_q^2\tau_q)$.
\end{proof}

\section[\NoCaseChange{Space Curves}]{\texorpdfstring{\NoCaseChange{Space Curves}}{Space Curves}}\label{sec:spacecurves}
We denote by $\curve\subset \C^n$ an algebraic curve whose associated
ideal is generated by $\bmf = (f_1,\dots,f_{n-1})\in\mathbb{Q}[\bmx]$
with $f_i$ of magnitude bounded by $(d, h)$. For $1 \le i \le n$, we denote by
$\pi_i$ the projection $(\xi_1, \ldots, \xi_n)\to (\xi_1, \ldots,
\xi_i)$.
Let $\operatorname{app}(\mathcal{C}_2)
:= \operatorname{sing}(\mathcal{C}_2)\setminus
\pi_{2}\bigl(\operatorname{sing}(\mathcal{C})\bigr)$
be the set of apparent singularities of $\mathcal{C}_2$.
Let $\calP \subset \operatorname{reg}(\mathcal{C})$ be a finite set
and $\calP_2 = \pi_2(\calP)$.
%
%
%
%
%
%
%
\textit{We assume that $(\curve, \calP)$ are in generic position},
in the sense that they satisfy properties $(H)$ 
of \cite{IslamPoteauxPrebet2023} which can be recovered up to a
generic linear change of coordinates 
\cite[Prop.~2.5]{IslamPoteauxPrebet2023}.
Then, the apparent singularities of the projected curve
$\mathcal{C}_2$ can be explicitly determined~\cite[Prop.~3.2]{IslamPoteauxPrebet2023}. 
Let $q_\app \in \Q[x_1]$ be the output of
the routine 
$\textsc{ApparentSingularities}$ given in
\cite[Prop. 5.1]{IslamPoteauxPrebet2023} which takes as input the parametrization 
$\curveparam$ of $\calC$ 
and returns a polynomial that defines the $x_1$-projections of the
apparent singularities of $\calC_2$. We denote by $\frakP$ a
zero-dimensional parametrization encoding $\calP$. 

\subsection{Graph constructions}
Let~$\mathcal{H}_w$ be the semi-algebraic isotopy graph of
$\mathcal{C}_2$ returned by
$\PlanarComponents\left(\omega, \frakP ,q_\app
\right) $.  The points at which the paramet\-rization of~$\mathcal{C}$
cannot be lifted to $\mathbb{R}^n$ are those cancelling $\partial w /
\partial x_2 = 0$.  All such points appear as
vertices of the graph~$\mathcal{H}_w$.  Then, the semi-algebraic
description of each edge of~$\mathcal{H}_w$ and of each vertex in
$V(\mathcal{H}_w)\setminus\mathcal{K}(\pi_1,\mathcal{C}_2)$ can be
lifted to~$\mathbb{R}^n$.  Also, since the restriction of $\pi_2$
to $\mathcal{C}$ fails to be injective at $x \in \mathcal{C}$
if and only if $\pi_2(x) \in \operatorname{app}(\mathcal{C}_2)$
\cite[Cor.~2.4]{IslamPoteauxPrebet2023}, the semi-algebraic
description of the vertices in
$\mathcal{K}(\pi_1,\mathcal{C}_2)\setminus
\operatorname{app}(\mathcal{C}_2)$ can be lifted to~$\mathbb{R}^n$ by
adding the equations $ \bmf = 0$ to their description in the
$(x_1,x_2)$-plane.

We need to determine which curve segments corresponding to
edges of~$\mathcal{H}_w$ belong to the same connected component of
$\mathcal{C}_{\mathbb{R}}$ and how to describe the missing points.
The first problem is addressed by the analysis 
in~\cite{IslamPoteauxPrebet2023}. It allows us to assemble the local
descriptions into the global structure of the connected components of
$\mathcal{C}_{\mathbb{R}}$.\\
\noindent
The points in $\operatorname{app}(\mathcal{C}_2)$ 
 are the only ones 
where the connectivity properties differ between $\mathcal{C}_2$ and
$\mathcal{C}$. The behavior of the curve in a neighborhood of such points follows
from~\cite[Lem.~4.5]{IslamPoteauxPrebet2023}.
Let $v_0$ be a vertex of the graph~$\mathcal{H}_w$ corresponding to a real point of
$\operatorname{app}(\mathcal{C}_2)$. Then exactly four edges of $\mathcal{H}_w$
emanate from $v_0$. Let $e_1,\dots,e_4$ be these edges, and let
$v_1,\dots,v_4$ be the vertices reached by following each edge $e_i$ away from
$v_0$.
The edges $e_1,\dots,e_4$ admit a unique cyclic ordering around $v_0$. 
With respect to this ordering, the edges
split into two pairs of opposite edges, $(e_1,e_3)$ and $(e_2,e_4)$. 
\cite[Lem.~4.5]{IslamPoteauxPrebet2023} implies
that, after lifting to $\mathbb{R}^n$, the branches of $\mathcal{C}_{\mathbb{R}}$
corresponding to each pair of opposite edges are connected in
$\mathcal{C}_\R$, while
branches corresponding to edges from different pairs are not.
Let $V_{\mathrm{app}} \subset V(\mathcal{H}_w)$ be the set of vertices
corresponding to
points of $\operatorname{app}(\mathcal{C}_2)$.
The procedure \textsc{SANodeResolution} takes as input
$(\mathcal{H}_w, V_{\mathrm{app}})$ and produces a new graph
$\mathcal{H}'$ gathering the vertices and edges of $\mathcal{H}_w$ lying
in the same connected components of $\curve_\R$, as does
\textsc{NodeResolution} \cite[Def. 4.6]{IslamPoteauxPrebet2023},  so that $\mathcal{H}'$
shares the same connectivity as $\curve_\R$ (see \cite[Prop.
4.7]{IslamPoteauxPrebet2023}). 
  The semi-algebraic data $\sigma_e$ associated with every edge is
  preserved, so that after calling
  \textsc{SANodeResolution}, we obtain, for each connected
  component $D_i$ of $\curve_\R$, semi-algebraic descriptions
  $\Theta_{i,1}, \ldots, \Theta_{i, \ell_i}$ such that their solution
  set coincides with $D_i\setminus \pi_2^{-1}(\app(\curve_2))$.  This
  is what we call a \textit{partial description of the connected
  components of $\curve_\R$}. Algorithm
  \PartialComponents takes as input $\curveparam$ and
  $\mathfrak{P}$ and returns such a partial description by calling
  successively \textsc{ApparentSingularities},
  \GenPlanarComponents and \textsc{SANodeResolution}
  (see the appendix).



\begin{lemma}\label{lemma:SAS3}
 Algorithm
 \PartialComponents is correct and uses 
 $\compsimpplanargraph$ bit
 operations. 
\end{lemma}

\begin{proof}
By \cite[Prop. 5.1]{IslamPoteauxPrebet2023}, 
\textit{ApparentSingularities} uses $\softO{\delta^6 + \delta^5\tau}$
bit operations and it outputs has magnitude $\left( \delta^2,
\softO{\delta^2 + \delta\tau}\right)$. Substituting this 
in \cref{thm:saisotopygraph-plane} yields a bit cost for Step
\ref{partial:plane} in $\compsimpplanargraph$. 

Correctness follows from the ones of
\textsc{ApparentSingularities}  (see
\cite{IslamPoteauxPrebet2023}), and 
\PlanarComponents 
(\cref{thm:saisotopygraph-plane}) and the
above discussion combined with the correctness of 
\textsc{NodeResolution}.
\end{proof}






\subsection{Algorithm description and correctness}

\noindent
Our main algorithm calls twice
\PartialComponents after applying two distinct linear changes of
coordinates, 
so that the semi-algebraic descriptions for the two
different systems of coordinates cover those points we were missing. 
For
$\bmA\in \GL_n(\C)$, we denote by $\curve_{2, \bmA}$ the Zariski
closure of $\pi_2(\curve^\bmA)$. 
%

\begin{theorem}\label{Twoparametrization}
Assume that the pair $(\mathcal{C},\mathcal{P})$ satisfies
condition~\emph{(H)}.
\noindent
There exists a non-empty Zariski open set
$\mathcal{A} \subset GL_n(\C)$
such that, for all $\bmA \in \mathcal{A}$, 
$(\curve^\bmA, \calP^\bmA)$ satisfies $(H)$
and the sets $\pi_2^{-1}\left(\app(\calC_2)  \right) $ and 
$\pi_2^{-1}\left(\app(\curve_{2,\bmA})\right))^{\bmA^{-1}} $ 
have an empty intersection.

\end{theorem}

The proof of \cref{Twoparametrization} is postponed to
\cref{sec:twoparams}.
It implies that, 
by choosing a second parametrization of $\mathcal{C}^\bmA$, for 
$\bmA\in \GL_n(\Q)$, we obtain, 
for each connected
  component $D_i$ of $\curve_\R$, semi-algebraic descriptions
  $\Xi_{i,1}, \ldots, \Xi_{i, s_i}$ such that their solution
  set coincides with $D_i\setminus \pi_2^{-1}(\app(\curve_2))^{\bmA^{-1}}$.
Since the failure sets are disjoint, combining the two descriptions
yields a semi-algebraic description of each connected component of
$\mathcal{C}_{\mathbb{R}}$. We explain now how to do that.

\noindent
Let $\mathfrak{T}$ encode a finite set
$\mathcal{T}\subset \reg(\mathcal{C})$ intersecting every connected
component of $\mathcal{C}_{\mathbb{R}}$.
Let $\mathcal{H}_1$ and $\mathcal{H}_2$ be the outputs of
\PartialComponents applied to
$(\curveparam,\mathfrak{T})$ and $(\curveparam^{\bmA},\mathfrak{T}^{\bmA})$,
respectively, where $\bmA\in \GL_n(\Q)$ satisfies 
the assumptions of Theorem~\ref{Twoparametrization}.
The procedure \ConnectParams takes as input
$\mathfrak{T}$, $\mathcal{H}_1$, and $\mathcal{H}_2$, and outputs a
semi-algebraic description of each connected component of
$\mathcal{C}_{\mathbb{R}}$.
It works as follows.
Let $\Theta_{i,1},\dots,\Theta_{i,t_i}$ be the semi-algebraic descriptions obtained
from $\mathcal{H}_1$ describing
$D_i\setminus \pi_2^{-1}(\app(\mathcal{C}_2))$.
If $D_i\setminus \pi_2^{-1}(\app(\mathcal{C}_2))$ reduces to a single singular point, 
then $t_i=1$ and $D_i$ is described by $\Theta_{i,1}$.
Otherwise, there exists $\zeta\in\mathcal{T}$ and
$\alpha\in\{1,\dots,t_i\}$ such that $\zeta\in\Theta_{i,\alpha}$.
We then select the semi-algebraic descriptions
$\Xi_{j,1},\dots,\Xi_{j,s_j}$ obtained from $\mathcal{H}_2$ whose descriptions
correspond to pieces of the same connected component of
$\mathcal{C}_{\mathbb{R}}$ and for which there exists
$\beta\in\{1,\dots,s_j\}$ such that $\zeta\in\Xi_{j,\beta}$.
The semi-algebraic description of $D_i$ is given by
$D_i \;=\;
\Theta_{i,1} \cup \cdots \cup \Theta_{i,t_i}
\cup \Xi_{j,1} \cup \cdots \cup \Xi_{j,s_j}$.
\begin{proposition}\label{lem:ConnectParametrizations}
Algorithm \ConnectParams is correct and performs
$\compconnectparametrizations$ bit operations.
Here, $\delta$ denotes a bound on the degrees of $\curveparam$ and $\curveparam^\bmA$,
and $\delta_{\mathcal{T}}$ for that of $\mathfrak{T}$.
\end{proposition}
\begin{proof}
Correctness follows from the correctness of
\PartialComponents and from Theorem~\ref{Twoparametrization},
which ensures the existence of $\zeta\in\cap D_i\setminus \pi^{-1}_2(\app(\calC_2))$ and the disjointness
of the failure sets, so that the union of the selected descriptions yields $D_i$.
Since the bulk of \ConnectParams consists in going
through the vertices and edges of $\mathcal{H}_1$ and
$\mathcal{H}_2$ (whose cardinalities lie in $O\left(\delta^4+\delta\delta_\mathcal{T} \right)$, 
see Theorem~\ref{thm:saisotopygraph-plane}), the number of bit operations 
is in $O\left(\delta^4+\delta\delta_\mathcal{T}\right)$.
\end{proof}
By \cite[Prop.~2.5]{IslamPoteauxPrebet2023}, there exists a non-empty
Zariski open subset
$\mathcal{U} \subset \mathrm{GL}_n(\mathbb{C})$
such that, for all $\bmA \in \mathcal{U}$,
$(\calC^\bmA,\querypoints^\bmA)$ satisfies~\emph{(H)}.
\begin{algorithm}[!htb]
\caption{\ComputeRealComponents}
\label{alg:compute-real-components}
\begin{algorithmic}[1]
\Require
$\bmf\in\mathbb{Q}[\bmx]$, $0 < \epsilon < 1$, 
$\bmA_1\in \mathcal{U}$ and $\bmA_2\in \mathcal{A}$
\Ensure
semi-algebraic
descriptions of the components of $\mathcal{C}_{\mathbb{R}}$, 

\State\label{connected:P} 
$\mathfrak{P} \leftarrow \RegularPoints(\bmf, \epsilon)$
\State\label{connected:R} 
$\curveparam \leftarrow \textsc{OneDimParam}(\bmf,
\epsilon)$ 
\State
Compute 
$(\mathcal{C}^{\bmA_1},\mathcal{P}^{\bmA_1})$

\State\label{connected:SAIG3}
$\mathcal{H}'_w \leftarrow
\PartialComponents(\mathcal R^{\bmA_1},\calP^{\bmA_1})$



\State
Compute
$(\mathcal{C}^{\bmA_2},\mathcal{P}^{\bmA_2})$

\State\label{connected:SAIG3bis}
$\mathcal{H}'_{w'} \leftarrow
\PartialComponents({\mathcal R}^{\bmA_2},\calP^{\bmA_2})$

\State\label{connected:connectparam} \Return
$(D_i)_{i=1}^r \leftarrow
\ConnectParams(\mathcal{H}'_w,\mathcal{H}'_{w'})$
\end{algorithmic}
\end{algorithm}

\begin{proof}[Proof of correctness Theorem \ref{thm:main}]
By~\cite[Prop.~2.5]{IslamPoteauxPrebet2023}, 
choosing $\bmA_1\in\mathcal{U}=GL_n(\Q)\setminus\mathscr{Z}$
ensures that the pair $(\mathcal{C}^{\bmA_1},\mathcal{P}^{\bmA_1})$
satisfies hypothesis~\textup{(H)}.
Then, using the notation in Theorem~\ref{Twoparametrization}, choosing
$\bmA_2\in\mathcal{A}=GL_n(\Q)\setminus\mathscr{Z}_{\bmA_1}$
guarantees that the two parametrizations have disjoint failure sets 
and that $(\mathcal{C}^{\bmA_2},\mathcal{P}^{\bmA_2})$
satisfies~\textup{(H)}.
Correctness of \PartialComponents and
\ConnectParams follows from Lemma \ref{lemma:SAS3} and 
Lemma \ref{lem:ConnectParametrizations}.
\end{proof}

\subsection{Subroutines}

\begin{lemma}\label{lemma:P}
  There exists an algorithm $\RegularPoints$ 
  whi\-ch, on input $\bmf$ and $0<\epsilon<1$
  computes 
  a zero-dimensional parametri\-zation for $\calP$ of magnitude
  $\left( O\left( n^3 d^n \right),\heightP \right) $ using $\totalcompP$
 bit operations with probability of
  success at least $1-\epsilon$.  
\end{lemma}
\begin{proof}
  Since $\langle \bmf\rangle $ is radical, at all points of
  $\reg(\curve)$, one of the $\left( n-1 \right)$-minors $\Delta_1, \ldots,
  \Delta_{n-1}$ of the Jacobian matrix $J(\bmf)$ associated to $ \bmf
  $ is non-zero.  Let $\Delta \in \mathscr{D} = \{\Delta_1, \ldots,
  \Delta_{n-1}\} $ and $T$ be a new variable. We compute one
  point in each connected component of the real algebraic set defined
  by $\bmf = T\Delta-1 = 0$ and iterate this when $\Delta$
  ranges over $\mathscr{D}$. Let $\bmg = (\bmf, \Delta)$. Note that
  the Jacobian matrix $J(\bmg)$ is full rank at
  any point of $V(\bmg)$. Hence, $\bmg$ generates a radical ideal and
  defines a smooth algebraic set. We use the variant of the
  probabilistic algorithm of \cite{SaSc03} presented in
  \cite{EGGSS2025}. It picks a sufficiently generic
  matrix $\bmM\in \GL_{n+1}(\Q)$ and solve the zero-dimensional
  systems $(1)$ $\bmg^\bmM = 0,\bm{\lambda} J( \bmg^\bmM ) = (1, 0,
  \ldots, 0)$ (where $\bm{\lambda} = (\lambda_1, \ldots, \lambda_n)$
  is a sequence of Lagrange multipliers) and $(2)$ $\bmg^\bmM = x_1 -
  \eta = 0$ for $\eta$ chosen sufficiently generic also. Both generate
  a radical ideal.  Using \cite{EGS23} as in \cite{EGGSS2025}), to
  ensure a probability of success greater than $\frac{1}{\epsilon}$,
  the bit sizes of $\eta$ and the entries of $\bmM$ are taken in
  $\softO{n+\log\left( \frac{1}{\epsilon} \right) }$.\\
  By the chain rule formula, solving $(1)$ is equivalent to
  solve $\bmg = 0$ and $\bm{\lambda} J(\bmg) = \bmm$ where $\bmm$ is
  the first column of $\bmM^{-1}$. As in \cite{EGGSS2025}, we use the
  symbolic homotopy algorithm of \cite{SaSc18}. Partitioning the
  variables as $\bmx, z, \bm{\lambda}$, one associates to each
  equation of $(1)$ their partial degrees in $\bmx, z, \bm{\lambda}$.
  It results that there are $(n-1)$ (resp. $n$) polynomials
  of tri-degree $\left( d, 0, 0 \right)$ (resp. $\left( n(d-1), 1, 1
  \right) $), $1$ of tri-degree $\left( n(d-1), 1, 0 \right)  $ and $1$
  of tri-degree $\left(n( d-1), 0, 1 \right)$. The bit size of their
  coefficients is bounded by $\tau \in h +
n(n+\log(\frac{1}{\epsilon}))$. Technical but easy
  calculations of multi-homogeneous B\'ezout degree and height bounds
  \cite[Prop. 3 \& 4]{SaSc18} show that this set has degree and
  height in $O\left( n^2 d^{n} \right) $ and $\heightP$  (see the appendix). 
Also, by
\cite[Lemma 5]{EGGSS2025}, (1) can be evaluated by a straight-line
program of length in $O\left( n^2 + n\binom{n+d}{d} \right)$ with
integer coefficients bounded by  $b\in \softO{h+n(
n+\log(\frac{1}{\epsilon}) )}$.  Let $\tau \in \softO{h+n(
n+\log(\frac{1}{\epsilon}) ) +d}$.  Applying \cite[Thm. 1]{SaSc18} as
in the proof of \cite[Thm.  1]{EGGSS2025}, one can solve $(1)$ using
$\compP$ bit operations.  \\ 
  Degree and height
bounds for $(2)$ are from B\'ezout-type bounds (still applying
\cite{SaSc18}) and smaller than the ones for $(1)$.
  Iterating this
  when $\Delta$ ranges over $\mathscr{D}$ multiplies by $n$ the 
  above estimates. Taking the union does not increase the complexity
  (see \cite[Lem.~J.3]{SaSc17}).
\end{proof}

\begin{lemma}\label{lemma:R}
  There exists an algorithm \textsc{OneDimParam} which, on input
  $\bmf$ and $0 <\varepsilon<1$ computes a one-dimensional
  parametrization of magnitude bounded by $\left( d^{n-1},
  \heightonedimparam \right)$ for $\curve$ using $\componedimparam$
  bit operations with probability $>1-\epsilon$. 
\end{lemma}

\begin{proof}
This is folklore, with close results stated in various places
(see \cite[Chap. 15, Thm. 18]{SchostPhD} for the first result of that
kind up to our knowledge). We follow the strategy introduced in
\cite{SchostPhD}. The degree bound is immediate by the
B\'ezout bound. By \cite[Prop. 5]{SaSc18}, computing a
zero-dimensional parametrization in some prime field (of sufficiently 
large characteristic)
for $V(\bmf)\cap \mathcal{H}$ where
$\mathcal{H}$ is a generic hypersurface is done in
$\softO{\log(\frac{1}{\epsilon})nd^{2(n-1)}\left(\binom{n+d}{n}+d+n^2\right)}$
arithmetic operations. Lifting this parametrization 
to obtain a one-dimensional
parametrization in the considered prime field requires 
$\softO{n\left( \binom{n+d}{n} +n^3 \right) d^{2n-1}}$
\cite[Lemma 3]{GiLeSa01}. 
It remains to lift this one-dimensional parametrization over the
rational numbers, using the bound $\tau = \softO{d^{n-1}(nh+n)}$, 
on the bit size of the output, 
given in \cite[Thm.~1]{DahanKadriSchost2018}. The conclusion follows
using again \cite[Lem. 3]{GiLeSa01}. 
\end{proof}

\subsection{Proof of \cref{thm:main} (Complexity)}
  The cost of Steps \ref{connected:P} and \ref{connected:R} are given
  by \cref{lemma:P,lemma:R}. 
  Let $\overline{\tau}$ be the maximum bit size of the entries of 
  the matrix $\bmA_1$ and $\bmA_2$. 
  By \cite[Lem.~J.1]{SaSc17} and \cite[Lem.~J.7]{SaSc17},
the linear change of variables induced by $\bmA_1$ and $\bmA_2$ 
on the 
parametrizations $\queryparam$ and $\curveparam$ use 
$\widetilde{O}\!\bigl((n^2 d_{\mathcal P}+n^2\delta^2+n^3)(\tau_{\mathcal P}+n\overline{\tau})\bigr)$ bit operations, without increasing
their degrees and with coefficients of bit size bounded by
$\tau_{\mathcal P}+n\overline{\tau}$ and $\tau+n\overline{\tau}$.
The cost of Steps \ref{connected:SAIG3} and \ref{connected:SAIG3bis}
are given by \cref{lemma:SAS3}. 

\noindent
Putting together all the previous steps, the number of bit operations 
performed by the algorithm 
\ComputeRealComponents is in {\small \totalcomplspace}
For $n\geq 3$, this is simplified to 
{\small $$\simplifiedtotalcompspace.$$}
\section[\NoCaseChange{Proof of \cref{Twoparametrization}}]{\texorpdfstring{\NoCaseChange{Proof of \cref{Twoparametrization}}}{Proof of \cref{Twoparametrization}}}\label{sec:twoparams}

By \cite[Prop.~2.5]{IslamPoteauxPrebet2023}, there exists a non-empty
Zariski open subset
$\mathcal{U} \subset \mathrm{GL}_n(\mathbb{C})$
such that, for all $A \in \mathcal{U}$, $(\calC^A,\mathcal P^A)$ satisfies~\emph{(H)}.
Let $\bmA \bm{x}=(\bm{a}_1, \ldots, \bm{a}_n)$ 
and $\pi_{\bma_1, \bma_2}$ be the projection $\xi\to \bma_1(\xi),
\bma_2(\xi)$. 

Let $\mathcal{C} \subset \mathbb{C}^n$ be an algebraic curve, such
that its associated ideal is generated by 
$\bm{f}=(f_1,\ldots,f_{n-1}) \subset \mathbb{Q}[\bmx]$.
Let $\curve_{2, \bmA}$ be the Zariski closure of
$\pi_{\bma_1, \bma_2}(\curve)$ and $w_\bmA$ be its defining polynomial. 
Let 
  $\frakA = \left( \fraka_{i,j} \right) $ be a $n \times n$ matrix
  with unknown entries and $\fraka_i$ be the $i$-th entry of $\frakA
  \bmx$ and 
$\zeta = (\zeta_1, \ldots, \zeta_n) \in
\operatorname{reg}(\mathcal{C})$.
\begin{lemma}\label{fixnode}
There exists a non-empty Zariski open set
$\mathcal{A}_\zeta \subset GL_n(\mathbb{C})$
such that, for all $A \in \mathcal{A}_\zeta$,
$\pi_{\bm{a}_1,\bm{a}_2}(\zeta)$ is not a node of
$\mathcal{C}_{2,\bmA}$.
\end{lemma}
\begin{proof}
By the Jacobian criterion \cite[Thm. 16.19]{Eisenbud1995}, 
the Jacobian matrix
$J_f$ of $\bm{f}$ has rank $n-1$ at any point of $\reg(\curve)$.
Let $\Delta$ be a $(n-1)\times(n-1)$ minor $\Delta$ of $J_f$. 
\noindent
We consider the system: 
{\footnotesize 
\[
  S_i(\zeta, \Delta)  \left\{
\begin{array}{l}
  u\cdot \Delta(\bm{x})-1=0,t\cdot(x_i-\zeta_i)-1=0,\\[2pt]
  f_1(\bm{x})=\cdots=f_{n-1}(\bm{x})=0,\\[2pt]
  \fraka_1(\bm{x}-\zeta)=0,
  \fraka_2(\bm{x}-\zeta)=0. 
\end{array}
\right.
\]}
where the variables $t$ and $u$ are new. Let 
$\frakc_1$ and $\frakc_2$ the sequence
of (unknown) coefficients of the linear forms $\fraka_1$ and
$\fraka_2$. 
%
%
We now consider the Jacobian matrix $J_{\bm{f}}$ of these $n+3$ equations with respect to
the variables ordered as $u, t, \bm{x}, \frakc_1, \frakc_2$. 
It has the block form
{\footnotesize 
\begin{equation*}
\begin{aligned}
&\left(
\begin{array}{c|c|c|c|c}
\Delta(x) & 0 & u\,\nabla\Delta(x) & 0 & 0\\ \hline
0 & x_i-\zeta_i & t & 0 & 0\\ \hline
0 & 0 & J_f(x) & 0 & 0\\ \hline
0 & 0 & {\frakc_1}^{\mathsf T} & (x-\zeta)^{\mathsf T} & 0\\ \hline
0 & 0 & {\frakc_2}^{\mathsf T} & 0 & (x-\zeta)^{\mathsf T}
\end{array}
\right)
\end{aligned}
\end{equation*}}
\noindent
where $\nabla(\Delta(x))$ denotes the gradient of the polynomial $\Delta$ with
respect to the variables $\bm{x}$.
A routine check shows that this Jacobian matrix has rank $n+3$. 
By the Jacobian criterion \cite[Thm.~16.19]{Eisenbud1995}, we deduce
that $V(S_i(\zeta, \Delta))$ is either empty or has co-dimension $n+3$ (and then
dimension $2n-1$). Then, its projection $\pi_{\frakc}$ on the $\frakc_1, \frakc_2$
space has co-dimension at least $1$. In other words, there exists a
non-zero 
polynomial $\Psi_{i, \zeta, \Delta}\in \Q[\frakc_1, \frakc_2]$ such that if
$\Psi_{i, \zeta, \Delta}(\bmc_1, \bmc_2)\neq 0$, then
$\pi_{\frakc}^{-1}\left( \bmc_1, \bmc_2 \right) \cap V\left(
S_{i}(\zeta, \Delta) \right) = \emptyset$. We denote by
$\mathcal{B}_{i,\zeta,\Delta}\subset \GL_n(\C)$ the non-empty Zariski open set defined
by $\Psi_{i,\zeta,\Delta}\neq 0$ and by $\mathcal{B}_{\zeta}$ the
intersection of all $\mathcal{B}_{i,\zeta,\Delta}$ when $i$ ranges
over $\{1, \ldots, n\}$ and $\Delta$ ranges over the $\left( n-1
\right)\times \left( n-1 \right)$ minors of $J_{\bm{f}}$. Note that
when $\bmA\in \mathcal{B}_{\zeta}$, there is no point $z$ in $\reg\left(
\curve \right)\setminus \{\zeta\} $ such that $\bma_1(z) =
\bma_1(\zeta)$ and $\bma_2(z) = \bma_2(\zeta)$ ($\bma_1$ and $\bma_2$
denote the first and second entry of $\bmA\bmx$). \\
Now, take $z\in \sing(\curve)$. Since $\zeta\in \reg(\curve)$ by
assumption, we have $z\neq \zeta$. Hence, 
$\Phi_{z, \zeta} = \left( \fraka_1(z)-\fraka_1(\zeta)
 \right) \left( \fraka_2(z)-\fraka_2(\zeta) \right) \in \Q[\frakc_1,
 \frakc_2]$ is not identically zero. Let $\mathcal{B}_{z,\zeta}$ be
 the non-empty Zariski open set in $\GL_n(\C)$ defined by $\Phi_{z,
 \zeta}\neq 0$. Let $\mathcal{B}_{\zeta,\sing} = \cap_{z\in \sing\left( \curve
 \right) } \mathcal{B}_{z,\zeta}$.\\
 Take $\bmA\in \mathcal{A}_{\zeta} = \mathcal{U}\cap \mathcal{B}_{\zeta}\cap
 \mathcal{B}_{\zeta, \sing}$. 
 Note that $\curve_{2,\bmA} = \pi_{\bma_1, \bma_2}(\curve^\bmA)$ (by
 \emph{(H)}) and
 for all $z\in \curve$, there is no point such that $\bma_1(z) =
 \bma_1(\zeta)$ and $\bma_2(z) = \bma_2(\zeta)$, i.e.
 $\pi_{\bma_1, \bma_2}(\zeta)$ is not a node of $\curve_{2,\bmA}$. 
\end{proof}

 \begin{proof}[Proof of \cref{Twoparametrization}]
 Consider the finite set $\pi_2^{-1}(\app(\mathcal C_2))\subset \reg(\mathcal C)$.
Since the condition of being an apparent singularity is stronger than that of
being a node, it suffices to define
$\mathcal A := \bigcap_{\zeta}\mathcal A_{\zeta}$,
where each $\mathcal A_{\zeta}$ is given by Lemma~\ref{fixnode}.
The resulting set $\mathcal A$ is a non-empty Zariski open set.
\end{proof}

\balance
\bibliographystyle{plain}
\bibliography{bibliography}

\appendix 

\section[\NoCaseChange{Subroutines}]{\texorpdfstring{\NoCaseChange{Subroutines}}{Subroutines}}

\begin{algorithm}[H]
\caption{\textsc{SignSamples}}
\label{alg:compute-sign-samples}
\begin{algorithmic}[1]
\Require $f\in\mathbb{Q}[x_1, x_2]$ and polynomials
$w_1,\dots,w_{d_2}\in\mathbb{Q}[x_1]$.
\Ensure For each $k\in\{1,\dots,d\}$, a finite set $\mathcal{Q}_k\subset\mathcal{C}_{\mathbb R}$
together with a sign map $s_k:\mathcal{Q}_k\to\{-1,0,+1\}$.

\State $w_1' \;\leftarrow\; \mathrm{d}w_1/ \mathrm{d}x_1$

\State Compute $B>0$ bounding the real roots of $w_1,\dots,w_{d_2}$.

\State Choose $\theta_1,\theta_2\in\mathbb{Q}$ such that $\theta_1<-B$ and $\theta_2>B$

\State ($\mathcal{Q}_1,s_1) \;\leftarrow\; \textsc{AssignSigns}(1,w'_1)$

\For{\textbf{each} $k\in\{2,\dots,d_2\}$}
    \State $P_k' \;\leftarrow\; \mathrm{d}(w_1\,w_k)/\mathrm{d} x_1$
    \State ($\mathcal{Q}_k,s_k) \;\leftarrow\; \textsc{AssignSigns}(k,P'_k(x-\theta_1)(x-\theta_2))$

\EndFor

\State \textbf{return} $\bigl(\{\mathcal{Q}_k\}_{k=1}^{d_2},\ \{s_k\}_{k=1}^{d_2}\bigr)$
\end{algorithmic}
\end{algorithm}

\begin{algorithm}[!htb]
\caption{\textsc{PreCond}}
\label{alg:pre-conditioning}
\begin{algorithmic}[1]
\Require
$\bmf\in\mathbb{Q}[x_1, x_2]$ 
\Ensure $(f_1, \ldots, f_r) \subset \Q[x_1, x_2]$ such that $f = f_1
\cdots f_r$ and for all $1\leq i \leq r$, $f_i$ satisfies $\calS$. 

\For{$k = 1,\ldots,d_2 - 1$}
  \State $g \leftarrow \gcd\left( f, {\partial^k f} / {\partial
  x_2^k} \right)$
  \If{$\deg(g) >0$}
    \State \Return $\textsc{PreCond}(g) \cup
    \textsc{Precond}(f / g)$
  \EndIf
\EndFor
\State \Return $f$
\end{algorithmic}
\end{algorithm}

\begin{algorithm}[!htb]
\caption{\PartialComponents}
\begin{algorithmic}[1]
\Require
$\curveparam=(w,\rho_2,...,\rho_n)$ and
$\frakP$ a one- and zero-dimensional parametrization.
\Ensure
A partial description of the connected components of $\curve_\R$, outside
$\app(\curve_2)$.

\State
$q_{\mathrm{app}} \leftarrow
\textsc{ApparentSingularities}(\curveparam)$

\State\label{partial:plane}
$[\mathcal{H}_w, V_{\mathrm{app}}, V_{\mathcal{P}}]
\leftarrow
\GenPlanarComponents(w,\mathcal{P}_2,q_{\mathrm{app}})$

\State \Return
$\mathcal{H}'_w \leftarrow
\textsc{SANodeResolution}(\mathcal{H}_w,V_{\mathrm{app}})$

\end{algorithmic}
\end{algorithm}

\section[\NoCaseChange{Degree, height and complexity for \cref{lemma:P}}]{\texorpdfstring{\NoCaseChange{Degree, height and complexity for \cref{lemma:P}}}{Degree, height and complexity for Lemma~\ref{lemma:P}}}
We provide here the technical details
yielding the degree and height bounds in the proof of
\cref{lemma:P}. The system of polynomials in $\Q[\bmx, z,
\bm{\lambda}]$ with $\bmx = (x_1, \ldots, x_n)$ and $\bm{\lambda} =
\left( \lambda_1, \ldots, \lambda_n \right)$ we study is 
  $\bmg = 0, z\Delta - 1 = 0, \bm{\lambda}J(\bm{g}) = \bmm$
With respect to $\bmx, z, \bm{\lambda}$, the $n$ 
polynomials in $\bmg$ have tri-degree bounded by $(d, 0, 0)$ and the
next one has tri-degree bounded by $(n(d-1), 1, 0)$. The first $n$
equations from $\bm{\lambda}J(\bm{g}) = \bmm$ have tri-degree bounded
by $\left( n(d-1), 1, 1 \right) $ while the last one has tri-degree
bounded by $\left( n(d-1), 0, 1 \right) $. Recall that the bit size of
the coefficients of these polynomials is bounded by $b\in \softO{h+n(
n+\log(\frac{1}{\epsilon}) ) }$.  Let $\tau \in \softO{h+n(
n+\log(\frac{1}{\epsilon}) ) +d}$.

To estimate the degree and the height of the algebraic set 
defined by this system, we follow
\cite[Sec. 2]{SaSc18} and  
consider: 
{\small 
  \begin{align*}
    P_1 & =  \left(\tau \zeta + d\vartheta_1 \right)^{n-1}, \quad 
        P_2 = \left(\tau \zeta + n(d-1)\vartheta_1 +\vartheta_2 +\vartheta_3
    \right)^n &\\
    P_3 & = \left(\tau \zeta + n(d-1) \vartheta_1 + \vartheta_2
      \right),\quad 
         P_4  = \left(\tau \zeta + n(d-1) \vartheta_1 + \vartheta_3 \right) &
     \end{align*}}
and 
$\mathbf{P} = P_1 P_2 P_3 P_4 \mod \langle \zeta^2, \vartheta_1^{n+1}, \vartheta_2^2,
  \theta_3^{n+1} \rangle \in \Z[\zeta, \vartheta_1, \vartheta_2,
  \vartheta_3]$.
 
By \cite[Prop. 3]{SaSc18}, the degree of $V$ is bounded
by the sum of the coefficients of $\mathbf{P}(0, \vartheta_1,
\vartheta_2, \vartheta_3)$. Note that $P_1 P_2 P_3 P_4$ is homogeneous 
of degree $2n+1$. The only monomial of degree $2n+1$ which does not
lie in $\mathfrak{m} = \langle \vartheta_1^{n+1}, \vartheta_2^2,
  \theta_3^{n+1} \rangle$ is $\vartheta_1^n \vartheta_2
  \vartheta_3^n$. \\
  We deduce that $\mathbf{P}(0, \vartheta_1,
  \vartheta_2, \vartheta_3) = C \vartheta_1^n \vartheta_2
  \vartheta_3^n$ with $C$ to be determined. \\
  It holds that $P_1(0, \vartheta_1, \vartheta_2, \vartheta_3) =
  d^{n-1}\vartheta_1^{n-1}$. The contributing terms of 
  $P_2(0, \vartheta_1, \vartheta_2, \vartheta_3)$ are of the
  form {\small $\binom{n}{a_2, b_2, c_2}\left(n(d-1)\right
    )^{a_2}\vartheta_1^{a_2} \vartheta_2^{b_2} \vartheta_3^{c_2}$}
  with $0 \le a_2 \le 1$, $0 \le b_2 \le 1$ and $0 \le c_2 \le n$ and
  $a_2+b_2+c_2 = n$. \\
  Taking into account $P_3$ and $P_4$, we deduce that $\mathbf{P}(0,
\vartheta_1, \vartheta_2, \vartheta_3)$ equals
{\small 
\begin{align*}
  d^{n-1}\sum_{}\binom{n}{a_2, b_2, c_2}\left(n(d-1)\right
    )^{a_2 + 2}\vartheta_1^{a_2 + n + 1} \vartheta_2^{b_2} \vartheta_3^{c_2}
    +\\[-0.5em]
  d^{n-1}\sum_{}\binom{n}{a_2, b_2, c_2}\left(n(d-1)\right
    )^{a_2 + 1}\vartheta_1^{a_2 + n} \vartheta_2^{b_2}
    \vartheta_3^{c_2 + 1}
    +\\[-0.5em]
  d^{n-1}\sum_{}\binom{n}{a_2, b_2, c_2}\left(n(d-1)\right
    )^{a_2 + 1}\vartheta_1^{a_2 + n} \vartheta_2^{b_2 + 1} \vartheta_3^{c_2}
    +\\[-0.5em]
  d^{n-1}\sum_{}\binom{n}{a_2, b_2, c_2}\left(n(d-1)\right
    )^{a_2}\vartheta_1^{a_2 + n-1} \vartheta_2^{b_2 + 1} \vartheta_3^{c_2
    + 1}
\end{align*} 
}
$\mod \mathfrak{m}$, 
  again with $0 \le a_2 \le 1$, $0 \le b_2 \le 1$ and $0 \le c_2 \le
  n$ and $a_2+b_2+c_2 = n$. The first sum is $0$. In the second one,
  we necessarily have $a_2 = 0$, $b_2=1$ and $c_2=n-1$. In the third
  one, we have $a_2=0$, $b_2=0$ and $c_2=n$. 
  In the fourth one, we necessarily have
  $a_2=1$, $b_2=0$ and $c_2=n-1$. Hence, all in all, we have 
  {\small \begin{align*}
    C = & d^{n-1} n(d-1)\left( \binom{n}{0,1,n-1} +
    \binom{n}{0,0,n} + \binom{n}{1,0,n-1} \right) \\[-0.1em]
      = &
    d^{n-1}n(d-1)\left( 2n+1 \right) \in O(n^2 d^n)
  .
  \end{align*}}
To obtain the height, we compute the sum of the coefficients
of $\mathbf{P}$ \cite[Prop. 4]{SaSc18}. Since $P_1 P_2
P_3 P_4$ is homogeneous of degree $2n+1$, the terms contributing to
this sum are those whose associated monomial is one of the following:
$\vartheta_1^n\vartheta_2\vartheta_3^n$, 
$\zeta \vartheta_1^{n-1}\vartheta_2\vartheta_3^n$, 
$\zeta \vartheta_1^n\vartheta_3^n$, 
$\zeta \vartheta_1^{n}\vartheta_2\vartheta_3^{n-1}$. 

Note that $P_1 \mod \mathfrak{m} = d^{n-1} \vartheta_1^{n-1} + \tau
(n-1)d^{n-2} \zeta\vartheta_1^{n-2} $.  We also have $P_2 \mod
\mathfrak{m} = \sum_{} \tau^{e_2} \left (n(d-1)\right
  )^{a_2}\binom{n}{e_2, a_2, b_2, c_2}
  \zeta^{e_2}\vartheta_1^{a_2}\vartheta_2^{b_2}\vartheta_3^{c_2}$ 
  with $0\le e_2 \le  1$, $0 \le a_2 \le 2$, $0\le b_2 \le 1$, $0\le
  c_2 \le n$ and $e_2+a_2+b_2+c_2=n$. Moreover, since the degree of
  $P_3 P_4$ is at most one in $\vartheta_3$ and all monomials in
  $\mathbf{P}$ have degree $2n+1$, we deduce that $n-1 \le
  c_2 \le n$. 
  The coefficients of those with $e_2=0$ are bounded
  by $n^4 d^{a_2}$. The coefficients of those with $e_2=1$ are bounded
  by $\tau n^4d^{a_2} $. We deduce that the sum of the coefficients
  of $\mathbf{P}$ is bounded by the sum of the coefficients of 
  \begin{align*}
    &n^4d^{n-1}\left( d\vartheta_1^{n-1}+ \tau n \zeta
    \vartheta_1^{n-2}\right ) \\[-0.5em]
    &    \left(\sum 
    \left(\tau d^{a_2}\zeta\vartheta_1^{a_2}\vartheta_2^{b_2}\vartheta_3^{c_2} + 
    d^{a'_2}\vartheta_1^{a'_2}\vartheta_2^{b'_2}\vartheta_3^{c'_2} \right)
    \right)P_3P_4\mod \mathfrak{m}
  \end{align*}
  with $a_2+b_2+c_2=n-1$, $0 \le a_2 \le 2$, $0\le b_2 \le 1$, $n-1\le
  c_2 \le n$ and $a'_2+b'_2+c'_2=n-1$, $0 \le a'_2 \le 2$, $0\le b'_2
  \le 1$, $n-1\le c'_2 \le n$. In order to bound the sum of the coefficients of the polynomial above, we may
replace $P_3$ and $P_4$ by
$\bigl(\zeta + nd\,\vartheta_1 + \vartheta_2\bigr)$ and
$\bigl(\zeta + nd\,\vartheta_1 + \vartheta_3\bigr)$, respectively.
We then analyze the coefficients of the resulting polynomial by splitting the
contribution coming from the sum
$d\vartheta_1^{n-1} + \tau n \zeta \vartheta_1^{n-2}$.
We first consider the term $d\vartheta_1^{n-1}$.
In this case, we must have $a_2 \le 1$ and $a'_2 \le 1$.
Moreover, the product with the terms $nd\,\vartheta_1$ coming from $P_3$ or $P_4$
does not vanish modulo $\mathfrak{m}$ if and only if $a_2 = 0$ or $a'_2 = 0$.
It follows that all the coefficients arising from this contribution are bounded by
$\tau n^5 d^{n+1}$.
Applying an analogous reasoning to the polynomial obtained from the term
$\tau n \zeta \vartheta_1^{n-2}$, we see that all the coefficients of this second
polynomial are bounded by $\tau n^6 d^{n}$.
Finally, since the constraints on
$a_2, b_2, c_2, a'_2, b'_2, c'_2$ imply that the number of terms in
$\sum \bigl(
\zeta \vartheta_1^{a_2}\vartheta_2^{b_2}\vartheta_3^{c_2}
+
\vartheta_1^{a'_2}\vartheta_2^{b'_2}\vartheta_3^{c'_2}
\bigr)$
is bounded by $24$, we conclude that the sum of the coefficients of $P$ lies in
$O\bigl(n^5 d^{n+1} + \tau n^6 d^{n}\bigr)\subset \softO{n^6d^{n+1} +
n^8d^n(h+\log\left( \frac{1}{\epsilon} \right) }$.

\noindent
  Finally, we apply \cite[Thm. 1]{SaSc18} (reusing the notation
  therein) with 
  {\small 
  \begin{align*}
    L \in O\left( n^2 +n\binom{n+d}{n} \right), \quad b\in \softO{
    h+n\left( n+\log(\frac{1}{\epsilon}) \right) }\\[-0.5em]
    \mathscr{C}_{\bm{n}}(\bm{d}) \in O\left( n^2d^n \right), \quad 
    \mathscr{H}_{\bm{n}}(\bm{\eta},\bm{d})\in
    \heightP\\[-0.5em]
    N = n+1, \quad \mathfrak{d} \le nd, \quad \bm{\eta} =
    (\tau,\ldots,\tau) \text{ and }
    s \leq \tau.  
\end{align*}}
  Since for $d\geq  2$, it holds that 
  $\binom{n+d}{n} \ge n^2$, we have $L\in O\left(
  n\binom{n+d}{n}\right)$. 
  We deduce that 
  $Lb \in n^2 h + n^3(n+\log(\frac{1}{\epsilon}))$ and that 
  $ \mathscr{C}_{\bm{n}}(\bm{d})\mathscr{H}_{\bm{n}}(\bm{\eta},\bm{d})$
  lies in $O\left(n^7 d^{2n+1} +\tau n^8 d^{2n}\right) \subset
  \softO{ n^{10} d^{2n+1} \left( h+\log(\frac{1}{\epsilon}) \right) }
  $.
  All in all, we obtain a number of bit operations lying in 
  $$\softO{n^{14} d^{2n+2} ( h+\log(\frac{1}{\epsilon})
  )  \binom{n+d}{n} }.$$
